\documentclass[10pt]{article}

\usepackage{booktabs} 
\usepackage{geometry}

\usepackage{url}
\usepackage{enumerate}
\usepackage{hyperref}

\usepackage{amsthm,amsfonts,amssymb,amsmath}
\usepackage{pifont}
\theoremstyle{definition}
\newtheorem{definition}{Definition}
\newtheorem{example}{Example}

\newtheorem{proposition}{Proposition}

\usepackage[ruled, vlined]{algorithm2e}
\usepackage{algorithmicx}

\usepackage{tabularx,multirow,ragged2e}
\newcolumntype{Y}{>{\RaggedRight\arraybackslash}X}
\newcolumntype{P}[1]{>{\RaggedRight\arraybackslash}p{#1}}

\newcommand{\ie}{\emph{i.e.~}}
\newcommand{\iev}{\emph{i.e.,~}}
\newcommand{\cf}{\emph{cf.~}}
\newcommand{\eg}{\emph{e.g.~}}
\newcommand{\al}{\emph{al.~}}
\newcommand{\espace}{\vspace{7pt}}
\newcommand{\seq}[1]{\boldsymbol{#1}}

\newcommand{\NegPSpan}{{\sc NegPSpan}}

\usepackage{calc}

\usepackage{mathtools}
\newcommand\myeq{\stackrel{\mathclap{\normalfont\mbox{\tiny def}}}{:=}}

\begin{document}
\title{NegPSpan: efficient extraction of negative sequential patterns with embedding constraints}

\author{Thomas Guyet -- Agrocampus-Ouest/IRISA UMR6074\\
Ren\'e Quiniou, Univ Rennes, Inria, CNRS, IRISA}
%



\maketitle
\begin{abstract}
Mining frequent sequential patterns consists in extracting recurrent behaviors, modeled as patterns, in a big sequence dataset.
Such patterns inform about which events are frequently observed in sequences, \ie what does really happen. Sometimes, knowing that some specific event does not happen is more informative than extracting a lot of observed events. 
Negative sequential patterns (NSP) capture recurrent behaviors by patterns containing both observed events and absent events. 
Few approaches have been proposed to mine such NSPs. In addition, the syntax and semantics of NSPs differ in the different methods which makes it difficult to compare them.
This article provides a unified framework for the formulation of the syntax and the semantics of NSPs. Then, we introduce a new algorithm, \NegPSpan, that extracts NSPs using a PrefixSpan depth-first scheme and enabling \emph{maxgap} constraints that other approaches do not take into account. The formal framework allows for highlighting the differences between the proposed approach and the methods from the literature, especially with the state of the art approach eNSP.
Intensive experiments on synthetic and real datasets show that \NegPSpan\ can extract meaningful NSPs and that it can process bigger datasets than eNSP thanks to significantly lower memory requirements and better computation times. 
\end{abstract}


\section{Introduction}
In many application domains such as diagnosis or marketing, decision makers show a  strong interest for  rules that associates specific events (a context) to undesirable events to which they are correlated or that are frequently triggered in such a context. Sequential pattern mining algorithms can extract such hidden rules from execution traces or transactions. in the classical setting, sequential patterns contain only positive events, \ie really observed events.
However, the absence of a specific action or event can often better explain the occurrence of an undesirable situation \cite{Cao2015}. 
For example in diagnosis, if some maintenance operations have not been performed, \eg damaged parts have not been replaced, then a fault will likely occur in a short delay while if these operations were performed in time the fault would not occur. In marketing, if some market-place customer has not received special offers or coupons for a long time then she/he has a high probability of churning while if she/he were provided such special offers she/he should remain loyal to her/his market-place.
In these two cases, mining specific events, some present and some absent, to discover under which context some undesirable situation occurs or not may provide interesting so-called \emph{actionable} information for determining which action should be performed to avoid the undesirable situation, \ie fault in diagnosis, churn in marketing.

We aim at discovering sequential patterns that take into account the absence of some events called \emph{negative events} \cite{Cao2015}. 
Moreover, we want to take into account some aspect of the temporal dimension as well, maximal pattern span or maximal gap between the occurrences of pattern events. 
For example, suppose that from a sequence dataset, we want to mine a sequential pattern $\seq{p} = \langle a\  b\rangle$ with the additional
\emph{negative} constraint telling that the event $c$ should not appear between events $a$ and $b$ in $\seq{p}$.
The corresponding negative pattern is represented as $\seq{p} = \langle a \  \neg c \  b\rangle$, where the logical sign $\neg$ denotes an absent event or set of events. 
Once the general idea of introducing negative statements in a pattern has been stated, the syntax and  semantics of such negative patterns should be clearly formulated since they have a strong impact both on algorithms outcome and their computational efficiency.
As we will see, the few algorithms from literature do not use the same syntactical constraints and rely on very different semantics principles (see Section \ref{sec:negpatterns:relwork}).  
More precisely, the efficiency of eNSP \cite{cao2016nsp}, the state-of-the-art algorithm for NSP mining, comes from a negation semantics that enables efficient operations on the sets of supported sequences. 
The two computational limits of eNSP are memory requirements and the impossibility for eNSP to handle embedding constraints such as the classical \emph{maxgap} and \emph{maxspan} constraints. 
When mining relatively long sequences (above 20 itemsets), such constraints appear semantically sound to consider short pattern occurrences where events are not too distant. In addition, such constraints can efficiently prune the occurrence search space. 

This article provides two main contributions:
\begin{itemize}
\item we clarify the syntactic definition of negative sequential patterns and we provide different negation semantics with their properties.
\item we propose \NegPSpan, an algorithm inspired by algorithm PrefixSpan to extract negative sequential patterns with \textit{maxgap} and \textit{maxspan} constraints. 
\end{itemize}
 
Intensive experiments compare, on synthetic and real datasets, the performance of \NegPSpan\ and eNSP as well as the pattern sets extracted by each of them.
We show that algorithm \NegPSpan\ is more time-efficient than eNSP for mining long sequences thanks to the maxgap constraint and that its memory requirement is several orders of magnitude lower, enabling to process much larger datasets.
In addition, we highlight that eNSP misses interesting patterns on real datasets due to semantic restrictions.


\section{Negative Sequential Patterns}
\label{sec:negpatterns}
This section introduces sequential patterns and negative sequential pattern mining. First we recall some basic definitions about sequences of itemsets, and classical sequential pattern mining then we introduce some definitions of negative sequential patterns.

In the sequel, $[n]=\{1, \dots, n\}$ denotes the set of the first $n$ strictly positive integers.
Let $(\mathcal{I}, <)$ be the set of items (alphabet) associated with a total order (\eg lexicographic order). 
An \emph{itemset} $A=\{a_1\ a_2\ ...\ a_m\}\subset \mathcal{I}$ is a set of ordered items.
A \emph{sequence} $\seq{s}$ is a set of sequentially ordered itemsets $\seq{s} = \langle s_1\ s_2\ ...\ s_n\rangle$. This means that for all $i,j \in [n],\; i < j$, $s_i$ appends before $s_j$ in sequence $\seq{s}$. This sequence starts by $s_1$ and finishes by $s_n$. 
Mining sequential patterns from a dataset of sequences, denoted $\mathcal{D}$, consists in extracting the frequent subsequences (patterns) included in database sequences having a support (\iev the number of sequences in which the pattern occurs) greater than a given threshold $\sigma$.
There is a huge literature about sequential pattern mining.
We will not go into details and refer the reader to a survey of the literature, such as Mooney et \al \cite{Mooney:2013}.

\espace

Negative sequential patterns (NSP) extend classical sequential patterns by enabling the specification of absent itemsets.
For example, $\seq{p} = \langle a\ b\ \neg c\ e\ f \rangle$ is a negative pattern.
The symbol $\neg$ before $c$ denotes that $c$ is a negative itemset (here reduced to an item).
Semantically, $\seq{p}$ specifies that items $a$ and $b$ happen in a row, then items $e$ and $f$ occur in a row, but item $c$ does not occur between the occurrences of $bs$ and $e$.

In the field of string matching, negation is classically defined for regular expression. 
In this case, a pattern is an expression that can hold any kind of negated \textit{pattern}. 
The same principle gives the following most generic definition of negative sequential patterns:
Let $\mathcal{N}$ be the set of negative patterns. 
A negative pattern $\seq{p} = \langle p_1\ \dots\ p_n \rangle \in \mathcal{N}$ is a sequence where $\forall i,\;p_i$ is a positive itemset ($p_i \subset \mathcal{I}$) or a negated pattern ($p_i=\neg \{q_i\},\; q_i\in \mathcal{N}$).

Due to its infinite recursive definition, $\mathcal{N}$ appears to be too huge to be an interesting and tractable search space for pattern mining. 
For instance, with $\mathcal{I}=\{a,b,c\}$, it is possible to express simple patterns like $\langle a\ \neg b\ c \rangle$ but also complex patterns like $\left\langle a,\neg \left\langle b,c \right\rangle \right\rangle$. The combinatorics for such patterns is infinite.

%

\bigskip

\label{sec:negpatterns:definitions}
We now provide our definition of negative sequential patterns (NSP) which introduces some syntactic restriction compare to the most generic case. These simple restrictions are broadly used in the literature \cite{Kamepalli:2014:FrequentNS} and enable us to propose efficient algorithms.

\begin{definition}[Negative sequential patterns (NSP)]\label{def:negativepattern}
A negative pattern $\seq{p} = \langle p_1\,\dots\ p_n\,\rangle$ is a sequence where $\forall i,\;p_i$ is a positive itemset ($p_i=\{p_i^j\}_{j\in[m]},\; p_i^j \in\mathcal{I}$) or a negated itemset ($p_i=\neg \{q_i^j\}_{j\in[m']},\; q_i^j\in \mathcal{I}$) under the two following constraints: consecutive negative itemsets and
negative itemsets at the pattern boundaries are forbidden. 
The \emph{positive part}\footnote{Called the \emph{maximal positive subsequence} in PNSP and Neg-GSP or the \emph{positive element id-set} in eNSP.} of pattern $\seq{p}$, denoted $\seq{p}^+$, is the subsequence of $\seq{p}$ restricted to its positive itemsets.
\end{definition}

According to the constraint of non consecutive negative itemsets, a negative pattern $\seq{p}$ can be denoted by
$\seq{p} = \langle p_1\ \neg q_1\ p_2\ \neg q_2\ \dots p_{n-1}\ \neg q_{n-1}\ p_{n}\rangle$ where $\forall i,\; p_i \subseteq I\setminus \emptyset$ and $q_i \subseteq I$. With this notation, $\seq{p}^+ = \langle p_1\ p_2\ \dots p_{n-1}\ p_{n} \rangle$.

\espace

Let us illustrate our syntactic restrictions by some \textbf{contra-examples of patterns} that our approach does not extract:
\begin{itemize}
\item first of all, a pattern is a sequence of positives and negative itemsets. It is not possible to have patterns such as $\left\langle a,\neg \left\langle b,c \right\rangle \right\rangle$
\item then, successive negated itemsets are not allowed: $\left\langle a\ \neg b\ \neg c d \right\rangle$ is not possible.
\item finally, a pattern finishing or starting by a negated itemsets is also not allowed $\left\langle \neg b\ d \right\rangle$.
\end{itemize}

\subsection{Semantics of Negative Sequential Patterns}\label{sec:negpatterns:semantics}
The semantics of negative sequential patterns relies on \emph{negative containment}: a sequence $s$ supports pattern $p$ if $s$ contains a sub-sequence $s'$ such that every positive itemset of $p$ is included in some itemset of $s'$ in the same order and for any negative itemset $\neg i$ of $p$, $i$ is \emph{not included} in any itemset occurring in the sub-sequence of $s'$ located between the occurrence of the positive itemset preceding $\neg i$ in $p$ and the occurrence of the positive itemset following $\neg i$ in $p$.

So far in the literature, the absence or non-inclusion of itemsets (represented here as a negative itemset) has been specified by loose formulations.
The authors of PNSP have proposed the set symbol $\nsubseteq$ to specify non-inclusion.
This symbol is misleading since it does not correspond to the associated semantics given in PNSP: an itemset $I$ is absent from an itemset $I'$ if the entire set $I$ is absent from $I'$ (as opposed to at least some item from $I$ is absent from $I'$) which corresponds to $I \cap I' =\emptyset$ in standard set notation, and not $I \not\subseteq I'$. We will call PNSP interpretation \emph{total non inclusion}. It should be distinguished from \emph{partial non inclusion} which corresponds (correctly) to the set symbol $\nsubseteq$. The symbol $\nsubseteq$ was further used by the authors of Neg-GSP and eNSP. The semantics of non inclusion is not detailed in Neg-GSP and one cannot determine if it means total or partial non inclusion.\footnote{Actually, though not clearly stated, it seems that the negative elements of Neg-GSP patterns consist of items rather than  itemsets. In this case, total and partial inclusion are equivalent.} eNSP does not define explicitly the semantics of non inclusion but, from the procedure used to compute the support of patterns, one can deduce that it uses total non inclusion.

\begin{definition}[non inclusion]\label{def:IS_notincluded}
We introduce two operators relating two itemsets $P$ and $I$:
\begin{itemize}
\item partial non inclusion: $P\not\preceq I \Leftrightarrow \exists e \in P$, $e \notin I$
\item total non inclusion: $P\not\sqsubseteq I \Leftrightarrow \forall e \in P, e \notin I$
\end{itemize}

Choosing one non inclusion interpretation or the other has consequences on extracted patterns as well as on pattern search. Let's illustrate this on related pattern support in the sequence dataset 
$$\mathcal{D} = \left\{
\begin{array}{l}
\seq{s}_1=\langle (bc)\ f\ a \rangle \\
\seq{s}_2=\langle (bc)\ (cf)\ a \rangle \\
\seq{s}_3=\langle (bc)\ (df)\ a \rangle \\
\seq{s}_4=\langle (bc)\ (ef)\ a \rangle\\
\seq{s}_5=\langle (bc)\ (cdef)\ a \rangle
\end{array}
\right\}.$$
Table \ref{tab:partial-total} compares the support of progressively extended patterns under the two semantics to show whether anti-monotonicity is respected or not.
Let's consider pattern $\seq{p}_2$ on sequence $\seq{s}_2$. Considering that the positive part of $\seq{p}_2$ is in $\seq{s}_2$, $\seq{p}_2$ occurs in the sequence iff $(cd)\not\subseteq (cf)$. In case of total non inclusion, it is false that $(cd)\not\sqsubseteq (cf)$ because of $c$ that occurs in $(cf)$, and thus $\seq{p}_2$ does not occur in $\seq{s}_2$. But in case of a partial non inclusion, it is true that $(cd)\not\preceq (cf)$, because of $d$ that does not occurs in $(cf)$, and thus $\seq{p}_2$ occurs in $\seq{s}_2$.

\begin{table}[tb]
\small\centering
\caption{Lists of supported sequences in $\mathcal{D}$ by negative patterns $\seq{p}_i$, $i=1..4$ under the total and partial non inclusion semantics. Every pattern has the shape $\langle a\ \neg q_i\ b\rangle$ where $q_i$ are itemsets such that $q_i \subset q_{i+1}$.}
\label{tab:partial-total}

    \begin{tabular}{lcc}
        \hline
        ~                               & partial & total \\ 
        ~                               & non inclusion & non inclusion \\
        ~ & $\not\preceq$ & $\not\sqsubseteq$\\
        \hline
        $\seq{p}_1 = \langle b \neg c a \rangle$      & $\{\seq{s}_1, \seq{s}_3, \seq{s}_4\}$           & $\{\seq{s}_1, \seq{s}_3, \seq{s}_4\}$ \\ 
        $\seq{p}_2 = \langle b \neg (cd) a \rangle$   & $\{\seq{s}_1, \seq{s}_2, \seq{s}_3, \seq{s}_4\}$ & $\{\seq{s}_1, \seq{s}_4\}$ \\ 
        $\seq{p}_3 = \langle b \neg (cde) a \rangle$  & $\{\seq{s}_1, \seq{s}_2, \seq{s}_3, \seq{s}_4\}$   & $\{\seq{s}_1\}$ \\ 
        $\seq{p}_4 = \langle b \neg (cdeg) a \rangle$ &$\{\seq{s}_1, \seq{s}_2, \seq{s}_3, \seq{s}_4,\seq{s}_5\}$ & $\{\seq{s}_1\}$ \\ \hline
        ~                               & monotonic           & anti monotonic        \\
        \hline
    \end{tabular}
\end{table}

Obviously, partial non inclusion satisfies anti-monotonicity while total non inclusion does not.
In the sequel we will denote the general form of itemset non inclusion by the symbol $\nsubseteq$,  meaning either $\not\preceq$ or $\not\sqsubseteq$.
\end{definition}
 
Now, we formulate the notions of sub-sequence, non inclusion and absence by means of the concept of embedding.

\begin{definition}[positive pattern embedding]\label{def:positivepattern_embedding}
Let $\seq{s}=\langle s_1\,\dots\, s_n\rangle$ be a sequence and $\seq{p}=\langle p_1\,\dots\, p_m\rangle$ be a (positive) sequential pattern.
$\seq{e}=(e_i)_{i\in[m]}\in [n]^m$ is an \emph{embedding} of pattern $\seq{p}$ in sequence $\seq{s}$ iff $\forall i\in[m],\; p_i \subseteq s_{e_i}$ and $\forall i\in[m-1],\; e_{i}<e_{i+1}$
\end{definition}

\begin{definition}[Strict and soft embeddings of negative patterns]\label{def:NSP_embedding}
Let $\seq{s}=\langle s_1\,\dots\, s_n\rangle$ be a sequence and $\seq{p}=\langle p_1\,\dots\, p_m\rangle$ be a negative sequential pattern.

$\seq{e}=(e_i)_{i\in[m]}\in [n]^m$ is a \textbf{soft-embedding} of pattern $\seq{p}$ in sequence $\seq{s}$ iff $\forall i\in[m]$:
\begin{itemize}
\item $p_i \subseteq s_{e_i}$ if $p_i$ is positive
\item $p_i \nsubseteq s_j,\;\forall j\in [e_{i-1}+1,e_{i+1}-1]$ if $p_i$ is negative
\end{itemize}

$\seq{e}=(e_i)_{i\in[m]}\in [n]^m$ is a \textbf{strict-embedding} of pattern $\seq{p}$ in sequence $\seq{s}$ iff for all $i\in[m]$:
\begin{itemize}
\item $p_i \subseteq s_{e_i}$ if $p_i$ is positive
\item $p_i \nsubseteq \bigcup_{j\in [e_{i-1}+1,e_{i+1}-1]} s_j$ if $p_i$ is negative
\end{itemize}
\end{definition}

\begin{proposition}
\label{prop:sqsubset_eqembeddings}
\emph{soft}- and \emph{strict}-embeddings are equivalent when $\nsubseteq\myeq\not\sqsubseteq$.
\end{proposition}
\begin{proof}
see Appendix \ref{sec:proofs}.
\end{proof}

Let $\seq{p}^+=\langle p_{k_1}\,\dots\, p_{k_l}\rangle$ be the positive part of some pattern $\seq{p}$, where $l$ denotes the number of positive itemsets in $\seq{p}$. If $\seq{e}$ is an embedding of pattern $\seq{p}$ in some sequence $\seq{s}$, then $\seq{e}^+=\langle e_{k_1}\,\dots\, e_{k_l}\rangle$ is an embedding of the positive sequential pattern $\seq{p}^+$ in $\seq{s}$.

\espace

The following examples illustrate the impact of itemset non-inclusion operator and of embedding type.

\begin{example}[Itemset absence semantics]
Let $\seq{p}=\langle a\ \neg (bc)\ d\rangle$ be a pattern and four sequences:
\begin{center}
\begin{tabular}{lccc}
\hline 
Sequence & $\not\sqsubseteq$ & $\npreceq$ / strict-embedding & $\npreceq$ / soft-embedding \\ \hline
$\seq{s_1}=\langle a\ c\ b\ e\ d \rangle$ &\ding{51} & &\\
$\seq{s_2}=\langle a\ (bc)\ e\ d \rangle$ && &\\
$\seq{s_3}=\langle a\ b\ e\ d \rangle$ &\ding{51}&\ding{51} &\\
$\seq{s_4}=\langle a\ e\ d \rangle$ &\ding{51}&\ding{51} &\ding{51}\\
\hline
\end{tabular}
\end{center}

One can notice that each sequence contains a unique occurrence of $\langle a\ d \rangle$, the positive part of pattern $\seq{p}$. 
Using soft-embeddings and total non-inclusion ($\nsubseteq\myeq\not\sqsubseteq$), $\seq{p}$ occurs in $\seq{s_1}$, $\seq{s_3}$ and $\seq{s_4}$ but not in $\seq{s_2}$. 
Using the strict-embedding semantics and partial non-inclusion, $\seq{p}$ occurs in sequence $\seq{s_3}$ and  $\seq{s_4}$ considering that items $b$ and $c$ occur between occurrences of $a$ and $d$ in sequences $1$ and $2$. 
%
With partial non inclusion ($\nsubseteq\myeq\not\preceq$) and either type of embeddings, the absence of an itemset is satisfied if any of its item is absent. As a consequence,  $\seq{p}$ occurs only in sequence $\seq{s_4}$. 
\end{example}

Another point that determines the semantics of negative containment concerns the multiple occurrences of some pattern in a sequence: should every or only one occurrence of the pattern positive part in the sequence satisfy the non inclusion constraints? This point is not discussed in previous propositions for negative sequential pattern mining. Actually, PNSP and Neg-GSP require a weak absence (at least one occurrence should satisfy the non inclusion constraints) while eNSP requires a strong absence (every occurrence should satisfy non inclusion constraints). 

\begin{definition}[Negative pattern occurrence] \label{def:neg_occurrence}
Let $\seq{s}$ be a sequence, $\seq{p}$ be a negative sequential pattern, and $\seq{p}^+$ the positive part of $\seq{p}$.
\begin{itemize}
\item Pattern $\seq{p}$ \emph{softly-occurs} in sequence $\seq{s}$, denoted $\seq{p} \preceq \seq{s}$, iff there exists at least one (strict/soft)-embedding of $\seq{p}$ in $\seq{s}$.
\item Pattern $\seq{p}$ \emph{strictly-occurs} in sequence $\seq{s}$, denoted $\seq{p} \sqsubseteq \seq{s}$, iff for any embedding $\seq{e}'$ of $\seq{p}^+$ in $\seq{s}$ there exists an embedding $\seq{e}$ of $\seq{p}$ in $\seq{s}$ such that $\seq{e}'=\seq{e}^+$.
\end{itemize}
\end{definition}

Definition \ref{def:neg_occurrence} allows for formulating two notions of absence semantics for negative sequential patterns depending on the occurrences of the positive part:
\begin{itemize}
\item \emph{strict occurrence}: a negative pattern $\seq{p}$ occurs in a sequence $\seq{s}$ iff there exists at least one occurrence of the positive part of pattern $\seq{p}$ in sequence $\seq{s}$ and \textbf{every} such occurrence satisfies the negative constraints,
\item \emph{soft occurrence}: a negative pattern $\seq{p}$ occurs in a sequence $\seq{s}$ iff there exists at least one occurrence of the positive part of pattern $\seq{p}$ in sequence $\seq{s}$ and \textbf{one} of these occurrences satisfies the negative constraints.
\end{itemize}

\begin{example}[Strict vs soft occurrence semantics]
Let $\seq{p}=\langle a\ b\ \neg c\  d\  \rangle$ be a pattern and  $\seq{s_1}=\langle a\ b\ e\ d \rangle$ and $\seq{s_2}=\langle a\ b\ c\ a\ d\ e\ b\ d \rangle$ be two sequences. The positive part of $\seq{p}$ is $\langle a\ b\ d \rangle$. It occurs once in $\seq{s_1}$ so there is no difference for occurrences under the two semantics. But, it occurs thrice in $\seq{s_2}$  with embeddings $(1,2,5)$, $(1,2,8)$ and $(4,7,8)$. The two first occurrences do not satisfy the negative constraint ($\neg c$) while the second occurrence does.
Under the soft occurrence semantics, pattern $p$ occurs in sequence $\seq{s_2}$ whereas under the strict occurrence semantics it does not.
\end{example}

We also introduce \textbf{constrained negative sequential patterns}. We consider the two most common anti-monotonic constraints on sequential patterns:  \textit{maxgap} ($\theta\in\mathbb{N}$) and \textit{maxspan} ($\tau\in\mathbb{N}$) constraints. These constraints impact NSP embeddings.
An embedding $\seq{e}$ of a pattern $\seq{p}$ in some sequence $\seq{s}$ satisfies the \textit{maxgap} (resp. \textit{maxspan}) constraint iff $\seq{e}^+=\{e_i,\dots,e_n\}$, the embedding of the positive part of $\seq{p}$ satisfies the constraint, \ie $\forall i\in[n-1],\; e_{i+1}-e_{i}\leq \theta$ (resp. $e_n- e_1\leq \tau$).

\espace

The definitions of pattern support, frequent pattern and pattern mining task derives naturally from the notion of occurrence of a negative sequential pattern, no matter the choices for embedding (soft or strict), non inclusion (partial or total) and occurrence (soft or strict). However, these choices concerning the semantics of NSPs impact directly the number of frequent patterns (under the same minimal threshold) and further the computation time. The stronger the negative constraints, the lesser the number of sequences that hold some pattern, and the lesser the number of frequent patterns.

\espace

Finally, we introduce a partial order on NSPs that is the foundation of our efficient NSP mining algorithm.

\begin{definition}[NSP partial order]\label{def:partialorder}
Let $\seq{p} = \langle p_1\ \neg q_1\ p_2\ \neg q_2\ \dots p_{k-1}\ \neg q_{k-1}\ p_{k}\rangle$ and $\seq{p}' =\langle p'_1\ \neg q'_1\ p'_2\ \neg q'_2\ \dots $ $ p'_{k'-1}\ \neg q'_{k'-1}\ p'_{k'}\rangle$ be two NSPs s.t. $\forall i \in [k],\; p_i\neq \emptyset$ and $\forall i\in [k'],\; p'_i\neq \emptyset$. By definition, $\seq{p}\lhd\seq{p}'$ iff $k\leq k'$ and:
\begin{enumerate}
\item $\forall i \in [k-1],\; p_i\subseteq p'_i$ and $q_i\subseteq q'_i$
\item $p_k\subseteq p'_k$
\item $k'\neq k \implies p_k\neq p'_{k}$ (non-reflexive)
\end{enumerate}
\end{definition}


Intuitively, $\seq{p}\lhd\seq{p}'$ if $\seq{p}$ is shorter than $\seq{p}'$ and the positive and negative itemsets of $\seq{p}$ are pairwise included into the itemsets of $\seq{p}'$, but, in case of extension by additional itemsets.
The classical pattern inclusion fails to be anti-monotonic \cite{zheng:2009:negative}, since the change of scope of negative itemsets. We illustrate what's happening on two examples.
Let first consider the case of an ending negated itemset illustrated by Zheng et \al. with patterns $\seq{p}'=\langle b\ \neg c\ a\rangle$ and $\seq{p}=\langle b\ \neg c\rangle$: removing the $a$ make the positive pattern less constraint (more frequent), but is extend the scope of the negative constraint. Negation are more constraint and the anti-monotonicity is lost. 
This specific case does not impact our framework as our definition of NSP (see Definition \ref{def:negativepattern}) does not allow ending negated itemsets. But let us now consider the patterns $\seq{p}'=\langle b\ \neg c\ d a\rangle$ and $\seq{p}=\langle b\ \neg c a\rangle$, and the sequences $\seq{s}=\langle b\ e\ d\ c\ a\rangle$. $\seq{p}'$ occurs in $\seq{s}$ but not $\seq{p}$ has the scope of the negated itemset $\neg c$ changed it was restricted to the interval between $b$ and $d$ occurrence for $\seq{p}'$, but between $b$ and $a$ for $\seq{p}$.

What is important in our partial order $\lhd$, is that the embedding of the positive pattern yields an embedding for $\seq{p}$ that imposes the negative constraints on the exact same scopes than negative constraints of $\seq{p}'$. Thanks to the anti-monotonicity of $\not\sqsubseteq$, additional itemsets in negative patterns leads to over constraints the sequence.
These remarks give some intuition behind the following anti-monotonicity property (Proposition \ref{prop:antimonotonic}). The formal proof of the proposition can be found in Appendix \ref{sec:proofs}.


\begin{proposition}[Anti-monotonicity of NSP support] 
\label{prop:antimonotonic}
The support of NSP is anti-monotonic with respect to $\lhd$ when $\nsubseteq\myeq\not\sqsubseteq$ and soft-occurrences ($\preceq$) are considered.
\end{proposition}

We can notice that while the strict occurrence semantic ($\sqsubseteq$) is used, $\lhd$ lost the anti-monotonicity. Considering $\seq{p}'=\langle a\ (bc)\ \neg c\ d\rangle$, $\seq{p}=\langle a\ b\ \neg c\ d\rangle$ and $\seq{s}=\langle a\ \mathbf{b}\ (bc)\ e\ d\rangle$, then it is true that $\seq{p}' \sqsubseteq \seq{s}$, but not that $\seq{p} \sqsubseteq \seq{s}$. In the second case, there are two possible embeddings and the second one (which does not derived from the embedding of $\seq{p}'$) does not satisfy the negative constraint.

A second example illustrates another case that is encountered when a postfixed sequence restricts the set of valid embeddings: $\seq{p}'=\langle a\ \neg b\ d\ \mathbf{c}\rangle$, $\seq{p}=\langle a\ \neg b\ d\rangle$ and $\seq{s}=\langle a\ e\ d\ c\ b\ d\rangle$. Again, $\seq{p}'$ occurs only once while $\seq{p}$ occur twice and one of its embeddings does not satisfy the negated itemset.
This example shows that a simple postfix extension of NSP leads to loose the monotonicity property while the strict occurrence semantic is considered.

\section{Related Work} \label{sec:negpatterns:relwork}
Kamepalli et \al provide a survey of the approaches proposed for mining negative patterns \cite{Kamepalli:2014:FrequentNS}.
The three most significant algorithms appear to be PNSP, Neg-GSP and eNSP.
We briefly review each of them in the following paragraphs.

PNSP (Positive and Negative Sequential Patterns mining) \cite{hsueh:2008:PNSP} is the first algorithm proposed for mining full negative sequential patterns where negative itemsets are not only located at the end of the pattern.
PNSP extends algorithm GSP \cite{srikant1996:GSP} to cope with mining negative sequential patterns.
PNSP consists of three steps: i) mine frequent positive sequential patterns, by using algorithm GSP, ii) preselect negative sequential itemsets --- for PNSP, negative itemsets must not be too infrequent (should have a support less than a threshold \textit{miss\_freq}) --- iii) generate candidate negative sequences levelwise and scan the sequence dataset again to compute the support of these candidates and prune the search when the candidate is infrequent.
This algorithm is incomplete: the second parameter reduces the set of potential negative itemsets.
Moreover, the pruning strategy of PNSP is not correct \cite{zheng:2009:negative} and PNSP misses potentially frequent negative patterns.

Zheng et \al \cite{zheng:2009:negative} also proposed a negative version of algorithm  GSP, called Neg-GSP, to extract negative sequential patterns. They showed that traditional Apriori-based negative pattern mining algorithms relying on support anti-monotonicity have two main problems. The first one is that the Apriori principle does not apply to negative sequential patterns.
They gave an example of sequence that is frequent even if one of its sub-sequence is not frequent.
The second problem has to do with the efficiency and the effectiveness of finding frequent patterns due to a vast candidate space. 
Their solution was to prune the search space using the support anti-monotonicity over positive parts. 
This pruning strategy is correct but incomplete and it is not really efficient considering the huge number of remaining candidates whose support has to be evaluated.
To improve the efficiency of their approach, the authors proposed an incomplete heuristic search based on Genetic Algorithm to find negative sequential patterns \cite{zheng2010efficient}. 
We will see in Section \ref{sec:negpatterns:semantics} that anti-monotonicity can be defined considering a partial order relation based on common prefixes that a enable to design a complete, correct and efficient algorithm.

eNSP (efficient NSP) has been recently proposed by Cao et \al \cite{cao2016nsp}. It identifies NSPs by computing only frequent positive sequential patterns and deducing negative sequential patterns from positive patterns. 
Precisely, Cao et \al showed that the support of some negative pattern can be 
computed by arithmetic operations on the support of its positive sub-patterns, 
thus avoiding additional sequence database scans to compute the support of negative patterns. However, this necessitates to store all the (positive) sequential patterns with their set of covered sequences (tid-lists) which may be impossible in case of big dense datasets and low minimal support thresholds.
This approach makes the algorithm more efficient but it hides some restrictive constraints on the extracted patterns. First, a frequent negative pattern whose so-called positive partner (the pattern where all negative events have been switched to positive) is not frequent will not be extracted.
Second, every occurrence of a negative pattern in a sequence should satisfy absence constraints. We call this \emph{strong absence semantics} (see Section \ref{sec:negpatterns:semantics}).
These features lead eNSP to extract less patterns than previous approaches.
In some practical applications, eNSP may miss potentially interesting negative patterns from the dataset.

The first constraint has been partly tackled by Dong et \al with algorithm eNSPFI, an extension of eNSP which mines NSPs from frequent and some infrequent positive sequential patterns from the negative border \cite{Gong2017eNSPFI}.
E-msNSP \cite{xu2017msnsp} is another extension of eNSP  which uses multiple minimum supports:
an NSP is frequent if its support is greater than a local minimal support threshold computed from the content of the pattern and not a global threshold as in classical approaches.
A threshold is associated with each item, and the minimal support of a pattern is defined from the most constrained item it contains. Such kind of adaptive support prevents from extracting some useless patterns still keeping the pattern support anti-monotonic. The same authors also proposed high utility negative sequential patterns based on the same principles \cite{xu2017HighUtilityNSP} and applied on smart city data \cite{xu2018efficient}. An alternative approach has been proposed by Lin consisting in mining high-utility itemsets with negative unit profits \cite{lin2016fhn} but is not applied on sequential patterns. 
It is worth noting that this algorithm relies basically on the same principle as eNSP and so, present the same drawbacks, heavy memory requirements, strong absence semantics for negation. 
F-NSP+ \cite{dong2018f} extends the eNSP algorithm to use bitmap representations of itemsets. Using bitmap representations enable to speed up the eNSP algorithm, thanks to very efficient set operation on bitmaps. The F-NSP algorithm has a poor memory usage, while F-NSP+, which adapts the bitmap size to the dataset, requires slightly less memory.

SAPNSP \cite{liu2015sapnsp} tackles the problem of large amount of patterns by selecting frequent negative and positive patterns that are actionable. Patterns are actionable while they conform to \textit{special rules}. 

NegI-NSP \cite{qiu2017negi} proposes additional syntactic constraints on negative itemsets and uses the same strategy as e-NSP.

\espace

\begin{table*}[tb]
\small
\caption{Comparison of negative pattern mining proposals. Optional constraints are specified in Italic.}
\label{NSP-comparison}
\begin{tabularx}{\textwidth}{|l|Y|Y|Y|Y|} 
\hline & \textbf{PNSP} \cite{hsueh:2008:PNSP} & \textbf{NegGSP} \cite{zheng:2009:negative} & \textbf{eNSP} \cite{cao2016nsp} & \textbf{\NegPSpan} \\ 
\hline 
\textbf{negative elements} &  itemsets &  items\textit{?} &  itemsets &  itemsets \\ 
\hline 
\textbf{itemsets} & $\not\preceq$? & $\not\preceq$ & $\not\sqsubseteq$ & $\not\sqsubseteq$ \\ 
\hline 
\textbf{embeddings} & strict & strict? & strict & strict/\textit{soft} \\ 
\hline 
\textbf{occurrences} & soft & soft & strict & soft \\ 
\hline 
\textbf{constraints on negative itemsets} & not too infrequent ($supp \leqslant less\_freq $) & frequent items & positive partner is frequent & frequent items, \textit{bounded size} \\ 
\hline 
\textbf{global constraints on patterns} & &  & positive part is frequent (second greater threshold) &\textit{maxspan}, \textit{maxgap}\\ 
\hline 
\end{tabularx}
\end{table*}

To conclude this section on formal aspects of negative pattern mining, we provide in Table \ref{NSP-comparison} a comparison of several negative sequential pattern mining approaches wrt several features investigated in this section.
%
%
It is also important to precise that not any semantics is ``more correct'' than another one. 
Its relevancy depends on the information the data scientists want to capture in its datasets, and the nature of the data at hand. 
In this work, one of our objective is to provide a sound and insightful framework about negative patterns to enable users to choose the tool to use and to make this choice according to the semantic of the negation they want to use. 
Execution time is obviously an important choice criteria but it must overtake by semantic choice to first provide interesting, intuitive and sound results.
%

\section{Algorithm \NegPSpan} \label{sec:censp}

In this section, we introduce algorithm \NegPSpan~ for mining NSPs from a sequence database under \emph{maxgap} and \emph{maxspan} constraints and under a weak absence semantics with $\nsubseteq\myeq\not\sqsubseteq$ for itemset inclusion.
As stated in proposition \ref{prop:sqsubset_eqembeddings}, no matter the embedding strategy, they are equivalent under strict itemset inclusion. 
Considering occurrences, \NegPSpan~ uses the soft-occurrence semantic: at least one occurrence of the negative pattern is sufficient to consider that it is supported by the sequence.

For computational reasons, we make an additional assumption on the admissible itemsets as negative itemsets. The negative itemsets are restricted to one element of some language $\mathcal{L}^-$ in order to cut the combinatorics of negative itemsets.  
In the algorithm \NegPSpan~ presented below, $\mathcal{L}^-=\left\{I=\{i_1,\dots,i_n\}| \forall k,\, supp(i_k)\geq \sigma\right\}$ denotes the set of itemsets that can be built from frequent items. 
But this set could also be user defined when the user is interested in some specific sets of non-occurring events.
For instance, $\mathcal{L}^-$ could be the set of frequent itemsets, which would be more restrictive than the set of itemsets made of frequent itemsets.

\begin{algorithm}[bp]
\footnotesize

\LinesNumbered

\SetKwInOut{Input}{input}
\SetKwData{break}{break}
\SetKwComment{Comment}{//}{}

\SetKwFunction{FRec}{\NegPSpan}
\SetKwFunction{match}{Match}
\SetKwFunction{output}{OutputPattern}
\SetKwProg{Fn}{Function}{:}{}
\SetKwFunction{PositiveComposition}{PositiveComposition}
\SetKwFunction{PositiveSequence}{PositiveSequence}
\SetKwFunction{NegativeSequence}{NegativeExtension}
\SetKwFunction{NegativeComposition}{NegativeComposition}

\Input{$\mathcal{S}$: set of sequences, $\seq{p}$: current pattern, $\sigma$: minimum support threshold,  $occs$: list of occurrences, $\mathcal{I}^f$: set of frequent items, $\theta$: maxgap, $\tau$: maxspan}
	\BlankLine
    \Fn{\FRec{$\mathcal{S}$, $\sigma$, $\seq{p}$, $occs$, $\mathcal{I}^f$, $\theta$, $\tau$}}{
		\Comment{Support evaluation of pattern $\seq{p}$}
        \If{$|occs|\geq \sigma$} {
        	\output{$p$, $occs$}\;
        }\Else {
        	\KwRet\;
        }
        
        \Comment{Positive itemset composition}
        \PositiveComposition{$\mathcal{S}$, $\sigma$, $\seq{p}$, $occs$, $\mathcal{I}^f$, $\theta$, $\tau$}\;
        
        \Comment{Positive sequential extension}
        \PositiveSequence{$\mathcal{S}$, $\sigma$, $\seq{p}$, $occs$, $\mathcal{I}^f$, $\theta$, $\tau$}\;
        
        \If{$|\seq{p}| \geq 2$ and $|\seq{p}_{|p|}|=1$} {
        	\Comment{Negative sequential extension}
			\NegativeSequence{$\mathcal{S}$, $\sigma$, $\seq{p}$, $occs$, $\mathcal{I}^f$, $\theta$, $\tau$}\;
        }
  	}  	
\caption{\NegPSpan: recursive function for negative sequential pattern extraction}
\label{algo:CeNSP-Rec}
\end{algorithm}

\subsection{Main Algorithm}
\NegPSpan~ is based on algorithm PrefixSpan \cite{pei2004mining:prefixspan} which implements a depth first search and uses the principle of database projection to reduce the number of sequence scans. 
\NegPSpan~ adapts the pseudo-projection principle of PrefixSpan which uses a projection pointer to avoid copying the data.
For \NegPSpan, a projection pointer of some pattern $\seq{p}$ is a triple $\langle sid, ppred, pos \rangle$ where 
$sid$ is a sequence identifier in the database, 
$pos$ is the position in sequence $sid$ that matches the last itemset of the pattern (necessarily positive)
and $ppred$ is the position of the previous positive pattern.

Algorithm \ref{algo:CeNSP-Rec} details the main recursive function of \NegPSpan~ for extending a current pattern $\seq{p}$. The principle of this function is similar to PrefixSpan. Every pattern $\seq{p}$ is associated with a pseudo-projected database represented by both the original set of sequences $\mathcal{S}$ and a set of projection pointers $occs$.
First, the function evaluates the size of $occs$ to determine whether pattern $p$ is frequent or not. If so, it is outputted, otherwise, the recursion is stopped because no larger patterns are possible (anti-monotonicity property).

Then, the function tries three types of pattern extensions of pattern $\seq{p}$ into a pattern $\seq{p}'$:
\begin{itemize}
\item the positive sequence composition ($\leadsto_c$) consists in adding  one item to the last itemset of $\seq{p}$ (following the notations of Definition \ref{def:partialorder}, the extension corresponds to the case of $\seq{p}'$ is the extension of $\seq{p}$ where $k'=k$, $\forall i\in[k-1],\, q_i=q'_i$ and $|p'_{k}|=|p_{k}|+1$),
\item the positive sequence extension ($\leadsto_s$) consists in adding a new positive singleton itemset at the end of $\seq{p}$ ($k'=k+1$, $\forall i\in[k-1],\, q_i=q'_i$ and $|p'_{k'}|=1$),
\item the negative sequence extension ($\leadsto_n$) consists in inserting a negative itemset between the positive penultimate itemset of $\seq{p}$ and the last positive itemset of $\seq{p}$ ($k'=k$, $\forall i\in[k-2],\, q_i=q'_i$, $|q'_{k-1}|=|q_{k-1}|+1$ and $p'_{k}=p_{k}$). In addition, NSP are negatively extended iff $|p_{k}|=1$ to prevent from redundant pattern generation (see section \ref{sec:redundancy}).
\end{itemize}

The negative pattern extension is specific to our algorithm and is detailed in the next section.
The first two extensions are identical to PrefixSpan pattern extensions, including their gap constraints management, \ie \emph{maxgap} and \emph {maxspan} constraints between positive patterns.

\begin{proposition}[]
\label{prop:complet_correct}
The proposed algorithm is correct and complete.
\end{proposition}

Intuitively, the algorithm is complete considering that the three extensions enables to generate any NSP. 
For instance, pattern $\langle a\ \neg e\ b\ (ce)\ \neg (bd)\ a\rangle$ would be evaluated after evaluating the following patterns: $\langle a \rangle \leadsto_s \langle a\ b\rangle \leadsto_n \langle a\ \neg e\ b\rangle \leadsto_s \langle a\ \neg e\ b\ c\rangle \leadsto_c \langle a\ \neg e\ b\ (ce)\rangle \leadsto_s \langle a\ \neg e\ b\ (ce)\ a\rangle \leadsto_n \langle a\ \neg e\ b\ (ce)\ \neg b\ a\rangle \leadsto_n \langle a\ \neg e\ b\ (ce)\ \neg (bd)\ a\rangle$. 
Secondly, according to proposition~\ref{prop:antimonotonic}, the pruning strategy is correct. 

\subsection{Extension of Patterns with Negated Itemsets}

\begin{algorithm}[tbp]
\footnotesize

\LinesNumbered

\SetKwInOut{Input}{input}
\SetKwData{break}{break}
\SetKwData{continue}{continue}
\SetKwComment{Comment}{//}{}

\SetKwFunction{FRec}{\NegPSpan}
\SetKwFunction{Fonction}{NegativeExtension}
\SetKwFunction{match}{Match}
\SetKwProg{Fn}{Function}{:}{}

\Input{$\mathcal{S}$: set of sequences, $\seq{p}$: current pattern, $\sigma$: minimum support threshold,  $occs$: list of occurrences, $\mathcal{I}^f$: set of frequent items, $\theta$: maxgap, $\tau$: maxspan}
	\BlankLine
    \Fn{\Fonction{$\mathcal{S}$, $\sigma$, $\seq{p}$, $occs$, $\mathcal{I}^f$, $\theta$, $\tau$}}{
        	
        \For{$it \in \mathcal{I}^f$ }{
        	
        	\If{$\seq{p}[-2]$ is pos} {
        		\Comment{Insert the negative item at the penultimate position}
        		$\seq{p}.insert(\neg it)$\; 
        	} \Else {
        		\If{$it>\seq{p}[-2].back()$} {
        			\Comment{Insert an item to the penultimate (negative) itemset}
        			$\seq{p}[-2].insert(\neg it)$\;
        		} \Else {
        			\continue\;
        		}
        	}
        		$newoccs \gets \emptyset$\;
        		\For{$occ \in occs$}{
        			$found \gets false$\;
   		     		\For{$ sp = [occ.pred+1,occ.pos-1]$} {
      	  				\If{$it \in \seq{s}_{occ.sid}[sp]$}{
        					$found \gets true$\;
        					\break\;
        				}
        			}
        			\If{!$found$} {
        				$newoccs \gets newoccs \cup \{occ\}$\;
        			} \Else {
        				\Comment{Look for an alternative occurrence}
        				$newoccs \gets newoccs \cup $ \match{$s_{sid}$, $\seq{p}$, $\theta$, $\tau$}\;
        			}
        		}
        		
        		\FRec{$\mathcal{D}$, $\sigma$, $\seq{p}$, $newoccs$, $\mathcal{I}^f$}\;
        		$\seq{p}[-2].pop()$\;
       	 	}
    }
\caption{\NegPSpan: negative extensions}
\label{algo:NegExt}
\end{algorithm}

Algorithm \ref{algo:NegExt} extends the current pattern $\seq{p}$ with negative items. It generates new candidates by inserting an item $it \in \mathcal{I}^f$, the set of frequent items. Let $\seq{p}[-2]$ and $\seq{p}[-1]$ denote respectively the penultimate itemset and the last itemset of $\seq{p}$. If $\seq{p}[-2]$ is positive, then a new negated itemset is inserted between $\seq{p}[-2]$ and $\seq{p}[-1]$. Otherwise, if $\seq{p}[-2]$ is negative, item $it$ is added to $\seq{p}[-2]$. To prevent redondant enumeration of negative itemsets, only items $it$ (lexicographically) greater than the last item of $\seq{p}[-2]$ can be added.

Then, lines 10 to 20, evaluate the candidate by computing the pseudo-projection of the current database. According to the selected semantics associated with $\not\sqsubseteq$, \ie total non inclusion (see Definition \ref{def:NSP_embedding}), it is sufficient to check the absence of $it$ in the subsequence included between the occurrences of positive itemsets surounding $it$. 
To achieve this, the algorithm checks the sequence positions in the interval $[occ.ppred+1, occ.pos-1]$.
If $it$ does not occur in itemsets from this interval, then the extended pattern occurs in the sequence $occ.sid$.
Otherwise, to ensure the completeness of the algorithm, another occurrence of the pattern has to be searched in the sequence (\cf \FuncSty{Match} function that takes into account gap constraints).

For example, the first occurrence of pattern $\seq{p}=\langle abc\rangle$ in sequence $\langle abecabc \rangle$ is $occ_p=\langle sid,2,4\rangle$. Let's now consider $\seq{p}'=\langle ab\neg ec\rangle$, a negative extension of $p$. The extension of the projection-pointer $occ_p$ does not satisfy the absence of $e$. So a new occurrence of $p$ has to be searched for. $\langle sid,6,7\rangle$, the next occurrence of $\seq{p}$, satisfies the negative constraint.
Then, \NegPSpan~ is called recursively for extending the new current pattern $\langle ab\neg ec\rangle$.

We can note that the gap constraints $\tau$ and $\theta$ does not explicitly appear in this algorithm (except while a complete matching is required), but it impact indirectly the algorithm by narrowing the possible interval of line 13.

\subsubsection{Extracting NSP without surrounding negations}\label{sec:nonsurrounding}

An option restricts negated item to be not surrounded by itemsets containing this item. 
This alternative is motivated by the objective to simplify pattern understanding. 
A pattern $\langle a\neg b b c \rangle$ may be interpreted as ``there is exactly one occurrence of $b$ between $a$ and $c$''.
But, this may also lead to redundant patterns: $\langle a b \neg b c\rangle$ matches exactly the same sequences than $\langle a\neg b b c \rangle$ (see section \ref{sec:redundancy}). 
This second restriction can be disabled in our algorithm implementation. If so and for sake of simplicity, we preferred to yield the pattern $\langle a b \neg b c\rangle$. 

The set of such NSP can be extracted using the same algorithm, simply changing the candidate generation in Algorithm \ref{algo:NegExt}, line 2 by $it \in \mathcal{I}^f\setminus ( \seq{p}[-1] \cup \seq{p}[-2])$. Items to add to a negative itemset are among frequent items except surrounding items.

\subsubsection[Extracting NSP with partial non inclusion]{Extracting NSP with partial non inclusion ($\not\subseteq\myeq \not\preceq$)}

Algorithm \ref{algo:NegExt_PartialNonInclusion} present the variant of the negative extension algorithm while the partial non-inclusion is used ($\not\subseteq\myeq \not\preceq$). 
The backbone of the algorithm is similar: a candidate pattern with a negated itemset at the penultimate position is generated and it assesses whether this candidate is frequent or not. It is done by checking the absence of the itemset $is$ in the itemsets of the sequence at positions defined by the occurrence. 
The test of line 7 assesses that it is false that $is \not\preceq \seq{s}_{occ.sid}[sp]$: $is$ is not partial non-included in one the itemset of the sequence iff $is$ is a subset of it.

On the contrary to the previous approach, candidate patterns are generated based on $\mathcal{L}^-$ the list of itemsets. It is not possible to build itemsets from the list of items because, using this non-inclusion semantic, the support is monotonic (and not anti-monotonic). 
The combinatorics of this variant is thus significantly higher in practice because all element of $\mathcal{L}^-$ would be evaluated.

\begin{algorithm}[tbp]
\footnotesize

\LinesNumbered

\SetKwInOut{Input}{input}
\SetKwData{break}{break}
\SetKwData{continue}{continue}
\SetKwComment{Comment}{//}{}

\SetKwFunction{FRec}{\NegPSpan}
\SetKwFunction{Fonction}{NegativeExtension}
\SetKwFunction{match}{Match}
\SetKwProg{Fn}{Function}{:}{}

\Input{$\mathcal{S}$: set of sequences, $\seq{p}$: current pattern, $\sigma$: minimum support threshold,  $occs$: list of occurrences, $\mathcal{I}^f$: set of frequent items, $\theta$: maxgap, $\tau$: maxspan}
	\BlankLine
    \Fn{\Fonction{$\mathcal{S}$, $\sigma$, $\seq{p}$, $occs$, $\mathcal{I}^f$, $\theta$, $\tau$}}{
        	
        \For{$is \in \mathcal{L}^-$ }{
        		\Comment{Insert the negative itemset at the penultimate position}
        		$\seq{p}.insert(\neg is)$\; 
        		$newoccs \gets \emptyset$\;
        		\For{$occ \in occs$}{
        			$found \gets false$\;
   		     		\For{$ sp = [occ.pred+1,occ.pos-1]$} {
      	  				\If{$is\subseteq \seq{s}_{occ.sid}[sp]$}{
        					$found \gets true$\;
        					\break\;
        				}
        			}
        			\If{!$found$} {
        				$newoccs \gets newoccs \cup \{occ\}$\;
        			} \Else {
        				\Comment{Look for an alternative occurrence}
        				$newoccs \gets newoccs \cup $ \match{$s_{sid}$, $\seq{p}$, $\theta$, $\tau$}\;
        			}
        		}
        		
        		\FRec{$\mathcal{D}$, $\sigma$, $\seq{p}$, $newoccs$, $\mathcal{I}^f$}\;
        		$\seq{p}[-2]=\emptyset$\;
       	 	}
    }
\caption{\NegPSpan: negative extensions with partial non-inclusion (alternative to Algorithm \ref{algo:NegExt})}
\label{algo:NegExt_PartialNonInclusion}
\end{algorithm}

\subsection{Redundancy avoidance}\label{sec:redundancy}
The \NegPSpan algorithm is syntactically non-redundant but can in practice generates patterns that are semantically redundant. 

The semantic redundancy appears for pairs of patterns like $\langle a\ \neg b\ b\ c\rangle$ and $\langle a\ b\ \neg b\ c \rangle$: there are syntactically different but match the exact same set of sequences. Semantically, such pattern could be interpreted as ``\textit{there is not much than one occurrence of $b$ between $a$ and $c$}''.
For such patterns, it is possible to avoid generating both efficiently. 
Our solution is to not generate candidate patterns with negative items that are in the last itemset. Thus, only $\langle a\ b\ \neg b\ c \rangle$ would be generated.
In Algorithm \ref{algo:NegExt} line 2, the list of frequent items $\mathcal{I}^f$ is then replaced by $\mathcal{I}^f\setminus \seq{p}[−1]$. 
But, this modification makes loose the completeness of the algorithm. In fact, the pattern $\langle a\ \neg b\ b\rangle$ is not generated neither its semantically equivalent pattern $\langle a\ b\ \neg b\rangle$ because of the syntactic constraint on NSP that can not end with a negative itemset. 
In practice, we do not manage this kind of redundancy or prefer the sound and correct option of not surrounding negative itemsets (see section \ref{sec:nonsurrounding}).

A syntactic redundancy is introduced by adding the extension by negative items. For instance, the pattern $\langle a\ \neg b\ (cd)\rangle$ may be reached by two distinct paths $p_1: \langle a\ c\rangle \leadsto_c \langle a\ (cd)\rangle \leadsto_n \langle a\ \neg b\ (cd)\rangle $ or $p_2: \langle a\ c\rangle \leadsto_n \langle a\ \neg b\ c\rangle \leadsto_c \langle a\ \neg b\ (cd)\rangle $. 
To solve this problem, the algorithm first specifies the negative itemsets as a composition of negative items and then to compose the last itemset with new items. This discard the path $p_1$. In Algorithm \ref{algo:CeNSP-Rec}, line 8 enables negative extension only if the last (positive) itemset is of size 1.


\subsection{Execution Example}
This section illustrates the execution of the algorithm on a small example. 
Let us consider the dataset of sequences of Table \ref{tab:dataset_example} and the minimal support threshold $\sigma=2$.
In this example, we consider the following negative patterns semantic: total non-inclusion and strong absence.
No gap constraints are considered ($\theta=\infty$ and $\tau=\infty$).
Then, we have $\mathcal{L}=\{a,b,c,d,e\}$. The $f$ event occurs only once and are thus not frequent according to $\sigma$ value.

\begin{table}[h!]
\centering
\caption{Dataset of sequences used in the execution example.}
\label{tab:dataset_example}
\begin{tabular}{ll}
\hline
SID & Sequence\\
\hline
$\seq{s}_1$ & $\langle a\ c\ b\ e\ d\rangle$ \\
$\seq{s}_2$ & $\langle a\ (bc)\ e\rangle$ \\
$\seq{s}_3$ & $\langle a\ b\ e\ d\rangle$ \\
$\seq{s}_4$ & $\langle a\ e\ d\ f\rangle$ \\
\hline
\end{tabular}
\end{table}

Figure \ref{fig:execution_example} illustrates the execution of \NegPSpan\ algorithm on the dataset of Table \ref{tab:dataset_example} starting from pattern $\langle a\rangle$.
The tree illustrates successive patterns explored by the depth-first search strategy.
Each node detailed both the pattern and the corresponding projected database. For sake of space, the tree is simplified and some nodes are missing.

For patterns larger than two, projected sequences have two colors corresponding, in green, of the part of the sequence that can be used to make positive extensions and, in red, of the part of the sequence that is used to assess absence of items for negative extension.
Two markers locate positions of the projection pointer. The second pointer is the same as the one computed by \textsc{PrefixSpan}.

Let us consider projected sequences of pattern $\langle ae \rangle$. 
In the first sequence, $d$ is green as the part of the sequence ending the sequence after position of $e$. 
$cb$ are in read because this events are inbetween occurrence of $a$ and $e$.
This pattern can be extended in two ways:
\begin{itemize}
\item with negative items among $\mathcal{L}\setminus\{a,e\}$ ($a$ and $e$ are removed if the restriction on second restriction is activated),
\item with positive items among items that are frequent in the green parts.
\end{itemize}

Considering extension of pattern $\langle a\ e \rangle$ with a negative item, \eg $\neg c$, each sequence whose red part contains the item is discarded, the others remains identical. 
Extension by $\neg c$ leads to pattern $\langle a\neg c\ e \rangle$ only for sequences $s_3$ and $s_4$.

The extension of pattern $\langle a\ e \rangle$ by a positive item follows the same strategy as \textsc{PrefixSpan}. In this case, the algorithm only explore extension by $d$ item and projected pointers are updated to reduce further scanning. 

\espace

Adding a new negative item while the penultimate item is negative, append it in the negative itemset. 
In case of pattern $\langle a\neg c\ e \rangle$, $d$ is the only candidate because $e$ is one of the surrounding events and $b$ is above $c$ is the lexicograpĥic order. With the total non-inclusion, we again simply have to discard sequences that contain the item $d$ within their red part.\\
We can see in the case of the $\langle a\neg d\ e\ d\rangle$ extension that all combinatorics of itemsets may quickly satisfy all negation constraints. This suggests first to carefully select the appropriate $\mathcal{L}$ and second to use a maximum size for negative itemsets to avoid pattern explosion.

Finally, we also notice that extensions with negative items are not terminal recusive steps.
Once negative items have been inserted, new positive items can be append to the pattern.
We encounter this case with pattern $\langle a\neg d\ e \rangle$ which is extended by pattern $d$.



\begin{figure}[tbp]
\centering
\includegraphics[width=\textwidth]{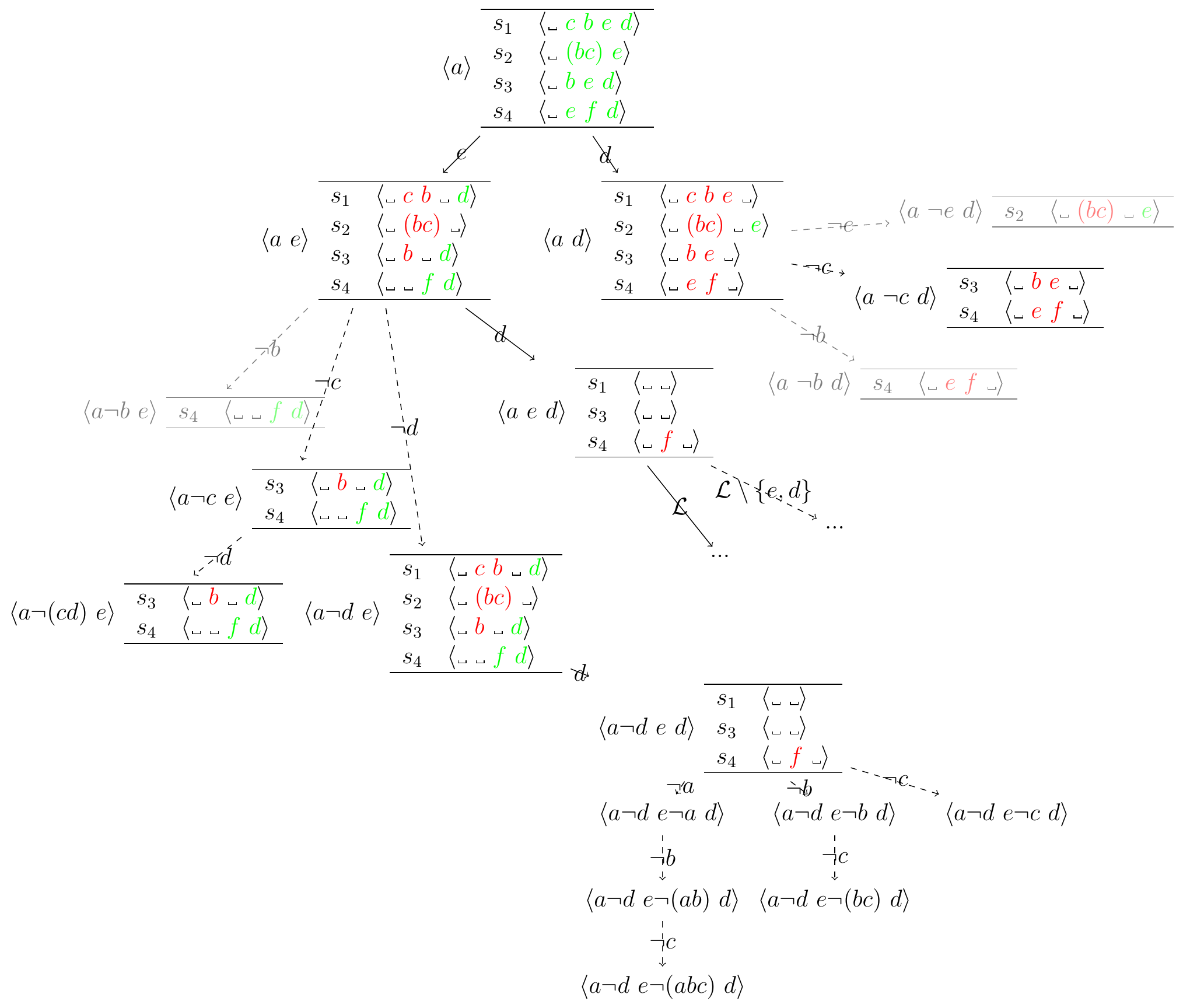}
\caption{Example of the tree search of the \NegPSpan\ algorithm of dataset of Table \ref{tab:dataset_example}.
Each node tree represents the pattern on the left and the projected database on the right. Itemsets are in red when the are used to assess future negations while green itemsets are used for sequential ($\leadsto_s$).
Dashed arrows represent negative extensions ($\leadsto_n$) while plain arrows are sequential ($\leadsto_s$) or compositional extensions ($\leadsto_c$). Arrow label holds the item that is used in the extension.
}
\label{fig:execution_example}
\end{figure}

\section{Experiments}\label{sec:expes}
This section presents experiments on synthetic and real data. Experiments on synthetic data aims at exploring and comparing \NegPSpan\ and eNSP for negative sequential pattern mining.
The other experiments were conducted on medical care pathways and illustrates results for negative patterns.
\NegPSpan\ and eNSP have been implemented in C++. 
We pay attention on the most significant results. More detailed results can be found in a compagnon website.\footnote{Code, data generator and synthetic benchmark datasets can be downloaded here: \url{http://people.irisa.fr/Thomas.Guyet/negativepatterns/}.
}

\subsection{Benchmark}
This section presents experiments on synthetically generated data. The principle of our sequence generator is the following: generate random negative patterns and hide or not some of their occurrences inside randomly generated sequences.
The main parameters are the total number of sequences ($n$, default value is $n = 500$), the mean length of sequences ($l = 20$), the number of different items ($d = 20$), the total number of patterns to hide ($3$), their mean length ($4$) and the minimum occurrence frequency of patterns in the dataset ($10\%$). 

Generated sequences are sequences of items (not itemsets). 
For such kind of sequences, patterns extracted by eNSP hold only items because positive partners have to be frequent. For a fair evaluation and preventing \NegPSpan\ from generating more patterns, we restricted $\mathcal{L}^-$ to the set of frequent items. 
For both approaches, we limit the pattern length to 5 items.

Figure \ref{fig:mg_cmp} illustrates the computation time and number of patterns extracted by eNSP and \NegPSpan\ on sequences of length 20 and 30, under three minimal thresholds ($\sigma=10\%$, $15\%$ and $20\%$) and with different values for the maxgap constraint ($\tau=4$, $7$, $10$ and $\infty$).
For eNSP, the minimal support of positive partners, denoted $\varsigma$, is set to 70\% of the minimal threshold $\sigma$. 
Each boxplot has been obtained with a 20 different sequence datasets.
Each run has a timeout of 5 minutes. 

\begin{figure*}[t]
\centering
\includegraphics[width=.44\textwidth]{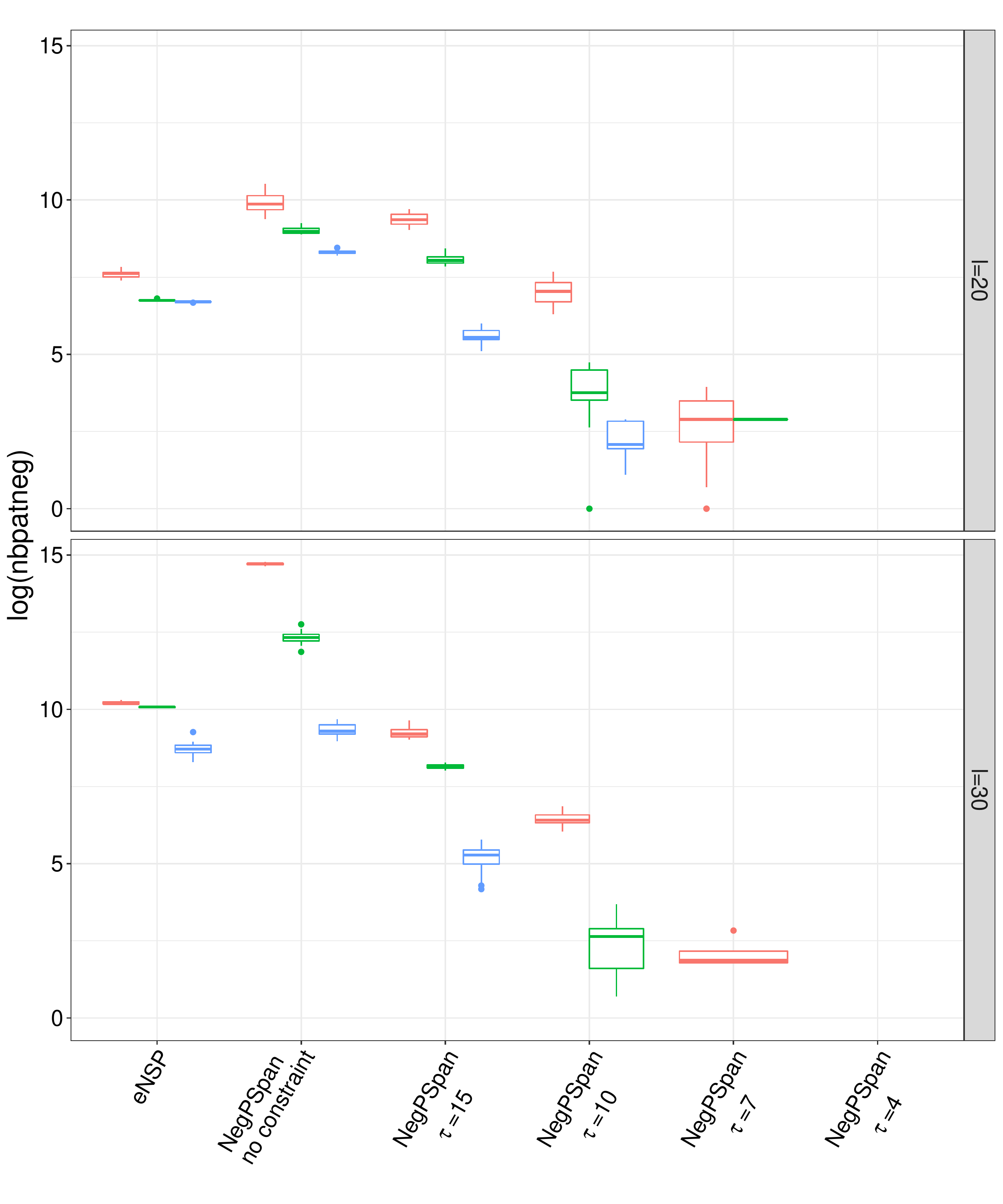} \hspace{0.1cm}
\includegraphics[width=.44\textwidth]{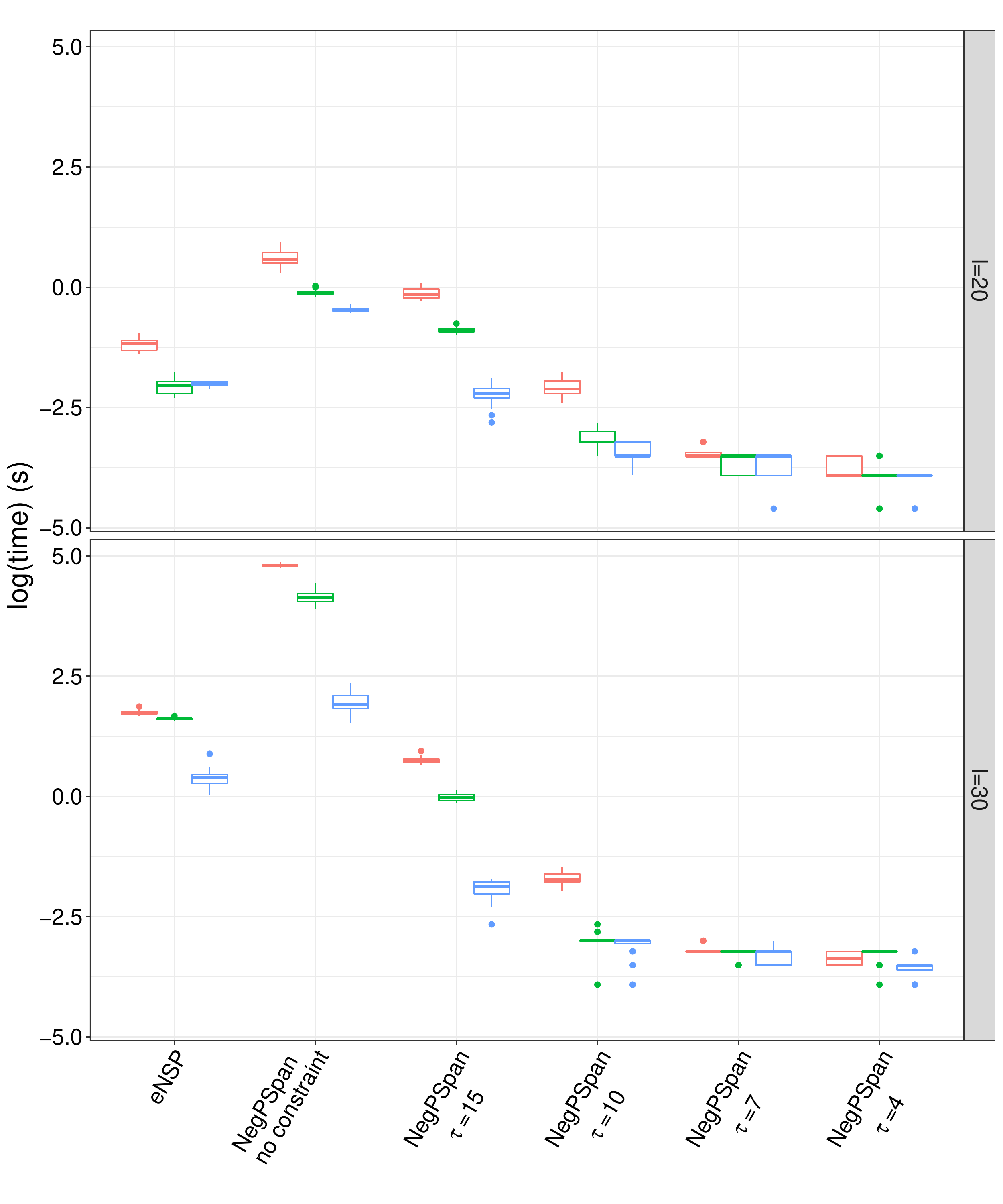} \hspace{0.1cm}
\includegraphics[width=.06\textwidth]{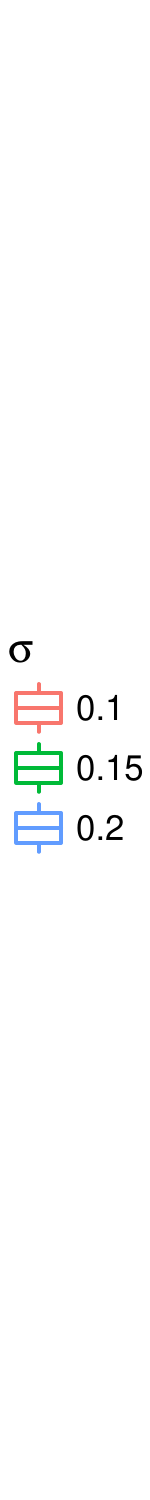}
\caption{Comparison of number of patterns (left) and computing time (right) between eNSP and \NegPSpan, with different values for maxgap ($\tau$). Top (resp. bottom) figures correspond to database with mean sequence length equal to $20$ (resp. $30$). Boxplot colors correspond to different values of $\sigma$ ($10\%$, $15\%$ and $20\%$).}
\label{fig:mg_cmp}
\end{figure*}

The main conclusion from Figure \ref{fig:mg_cmp} is that \NegPSpan\ is more efficient than eNSP when maxgap constraints are used.
As expected, eNSP is more efficient than \NegPSpan\ without any maxgap constraint.
This is mainly due to the number of extracted patterns. \NegPSpan\ extracts significantly more patterns than eNSP because of different choices for the semantics of NSPs.
First, eNSP uses a stronger negation semantics. Without maxgap constraints, the set of patterns extracted by \NegPSpan\ is a superset of those extracted by eNSP (see proof in Appendix \ref{sec:proofs_superset}). 

An interesting result is that, for reasonably long sequences ($20$ or $30$), even a weak maxgap constraint ($\tau=10$) significantly reduces the number of patterns and makes \NegPSpan\ more efficient.
$\tau=10$ is said to be a \textit{weak} constraint because it does not cut early the search of a next occurring item compare to the length of the sequence ($20$ or $30$).
This is of particular interest because the maxgap is a quite natural constraint when mining long sequences. It prevents from taking into account long distance correlations that are more likely irrelevant.
Another interesting question raised by this results is the real meaning of extracted patterns by eNSP.
In fact, under low frequency thresholds, it extracts numerous patterns that are not frequent when weak maxgap constraints are considered.
As a consequence, the significance of most of the patterns extracted by eNSP seems poor while processing ``long'' sequences datasets.

Figure \ref{fig:mg_cmp} also illustrates classical results encountered with sequential pattern mining algorithms. We can note that, for both algorithms, the number of patterns and runtime increase exponentially as the minimum support decreases. Also, the number of patterns and the runtime increase notably with sequence length.

\begin{figure*}[t]
\centering
\includegraphics[width=.45\textwidth]{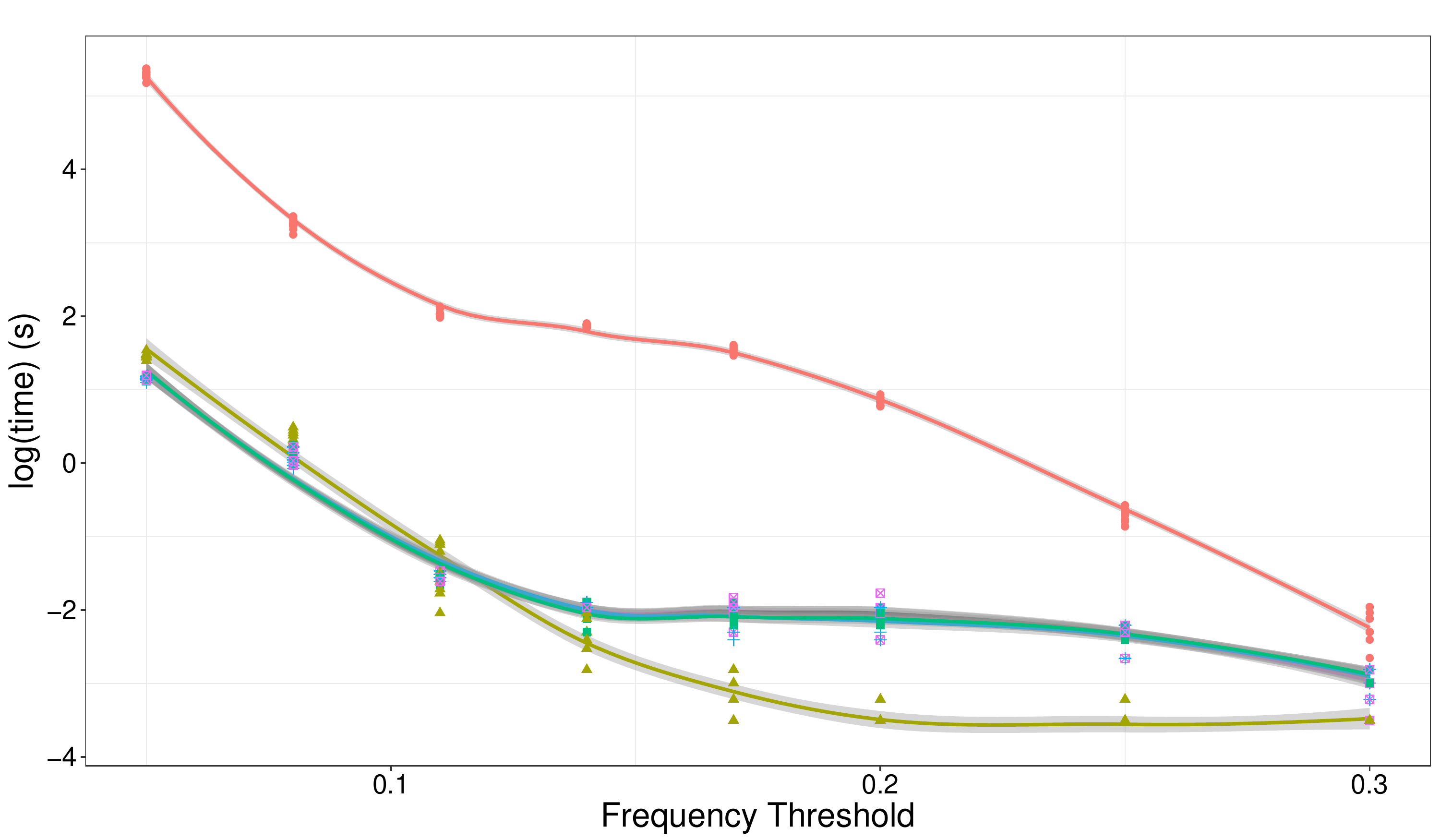}
\includegraphics[width=.45\textwidth]{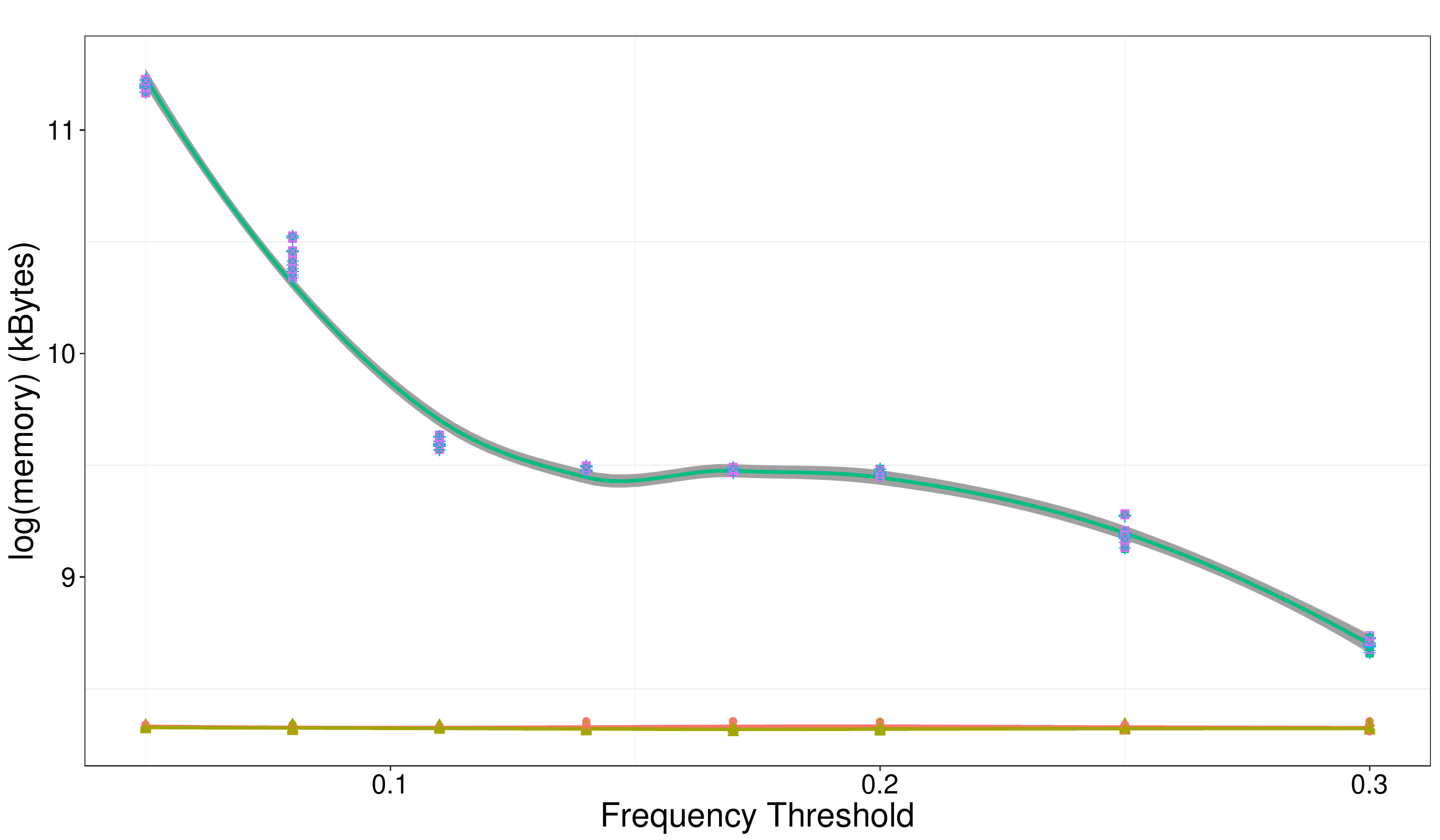}\\
\includegraphics[width=.45\textwidth]{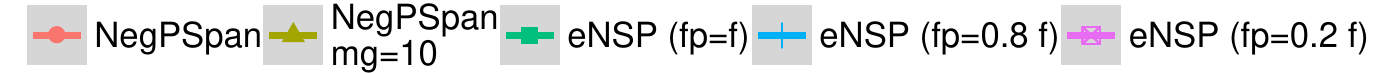}
\caption{Comparison of computing time (left) and memory consumption (right) between eNSP and \NegPSpan\ wrt minimal support.}
\label{fig:th_cmp}
\end{figure*}

\espace

Figure \ref{fig:th_cmp} illustrates computation time and memory consumption with respect to minimum threshold for different settings: eNSP is ran with different values for $\varsigma$, the minimal frequency of the positive partner of negative patterns ($100\%$, $80\%$ and $20\%$ of the minimal frequency threshold) and \NegPSpan\ is ran with a maxgap of $10$ or without.
Computation times show similar results as in previous experiments: \NegPSpan\ becomes as efficient as eNSP with a (weak) maxgap constraint.
We can also notice that the minimal frequency of the positive partners does not impact eNSP computing times neither memory requirements.

The main result illustrated by this Figure is that \NegPSpan\ consumes significantly less memory than eNSP. This comes from the depth-first search strategy which prevents from memorizing many patterns. 
On the opposite, eNSP requires to keep in memory all frequent positive patterns and their occurrence list.
The lower the threshold is, the more memory is required.
This strategy appears to be practically intractable for large/dense databases.

\subsection{Experiments on Real Datasets}
This section presents experiments on the real datasets from the SPMF repository.\footnote{\url{http://www.philippe-fournier-viger.com/spmf/index.php?link=datasets.php}} These datasets consist of click-streams or texts represented as sequences of items. Datasets features and results are reported in Table \ref{tab:realdatasets}. 
For every dataset, we have computed the negative sequential patterns with a maximum length of $l=5$ items and a minimal frequency threshold set to $\sigma=5\%$. \NegPSpan\ is set with a maxgap $\tau=10$ and eNSP is set up with $\varsigma=.7\sigma$. For each dataset, we provide the computation time, the memory consumption and the numbers of positive and negative extracted patterns.
Note that the numbers of positive patterns for eNSP are given for $\varsigma$ threshold, \ie the support threshold for positive partners used to generate negative patterns.

For the \textit{sign} dataset, the execution has been stopped after 10 mn to avoid running out of memory. 
The number of positive patterns extracted by eNSP considering the $\sigma$ threshold is not equal to \NegPSpan\ simply because of the maxgap constraint.
 
The results presented in Table \ref{tab:realdatasets} confirm the results from experiments on synthetic datasets. 
First, it highlights that \NegPSpan\ requires significant less memory for mining every dataset. 
Second, \NegPSpan\ outperforms eNSP for datasets having a long mean sequence length (\textit{Sign}, \textit{Leviathan}, and \textit{MSNBC}).
In case of the \textit{Bible} dataset, the number of extracted patterns by eNSP is very low compared to \NegPSpan\ due to the constraint on minimal frequency of positive partners.

\begin{table}[tb]

\footnotesize
\caption{Results on real datasets with setting $\sigma=5\%$, $l=5$, $\tau=10$, $\varsigma=.7 \sigma$. Bold faces highlight lowest computation times or memory consumptions.}
\label{tab:realdatasets}
~\hspace{-1.8cm}
\begin{tabular}{l|rrr|rrrr|rrrr|}
                      &  \multicolumn{3}{c|}{Dataset}  & \multicolumn{4}{c|}{\NegPSpan} & \multicolumn{4}{c|}{eNSP} \\ \hline
 & \multicolumn{1}{c}{$|\mathcal{D}|$} & \multicolumn{1}{c}{$|\mathcal{I}|$} & \multicolumn{1}{c|}{length} & \multicolumn{1}{c}{time ($s$)} & \multicolumn{1}{c}{mem ($kb$)} & \multicolumn{1}{c}{\#pos} & \multicolumn{1}{c|}{\#neg} & \multicolumn{1}{c}{time ($s$)} & \multicolumn{1}{c}{mem ($kb$)} & \multicolumn{1}{c}{\#pos} & \multicolumn{1}{c|}{\#neg}  \\ \hline
\textit{Sign}	& 730 &	267 &	51.99 &	\textbf{15.51} &	\textbf{6,220} &	348 &	1,357,278 &	349.84 (!) &	13,901,600 &	1,190,642 &	1,257,177  \\
\textit{Leviathan}	& 5,834	& 9,025	& 33.81	& \textbf{6.07}	& \textbf{19,932}	& 110	& 39797	& 28.43	& 428,916	& 7,691	& 17,220 \\
\textit{Bible}	& 36,369	& 13,905	& 21.64	& 38.82	& \textbf{68,944}	& 102	& 43,701	& \textbf{27.38}	& 552,288	& 1,364	& 2,621 \\
\textit{BMS1}	& 59,601	& 497	& 2.51	& \textbf{0.16}	& \textbf{22,676}	& 5	& 0	& 0.18	& 34,272	& 8	& 7\\
\textit{BMS2}	& 77,512	& 3,340	& 4.62	& 0.37	& \textbf{39,704}	& 1	& 0	& \textbf{0.35}	& 53,608	& 3	& 2\\
\textit{kosarak25k}	& 25,000	& 14804	& 8.04	& 0.92	& \textbf{24,424}	& 23	& 409	& \textbf{0.53}	& 43,124	& 50	& 51\\
\textit{MSNBC}	& 31,790	& 17	& 13.33	& \textbf{40.97}	& \textbf{41,560}	& 613	& 56,418	& 41.44	& 808,744	& 2,441	& 5,439\\ \hline
\end{tabular}
\end{table}

\subsection{Case Study: Care Pathway Analysis}
This section presents the use of NSPs for analyzing  epileptic patient care pathways. Recent studies suggest that medication changes may be associated with epileptic seizures for patients with long term treatment with anti-epileptic (AE) medication \cite{polard2015brand}. 
NSP mining algorithms are used to extract patterns of drugs deliveries that may inform about the suppression of a drug from a patient treatment.
In \cite{DauxaisAIME17}, we studied discriminant temporal patterns but it does not explicitly extract the information about medication absence as a possible explanation of epiletic seizures. 

Our dataset was obtained from the french insurance database \cite{Moulis2015411} called SNIIRAM. 8,379 epileptic patients were identified by their hospitalization
identified by their hospitalization related to an epileptic event.
For each patient, the sequence of drugs deliveries within the $90$ days before the epileptic event was obtained from the SNIIRAM. For each drug delivery, the event id is a tuple $\langle m, grp, g\rangle$ where $m$ is the ATC code of the active molecule, $g\in\{0,1\}$ is the brand-name ($0$) vs generic ($1$) status of the drug and $grp$ is the speciality group. The speciality group identifies the drug presentation (international non-proprietary name, strength per unit, number of units per pack and dosage form).
The dataset contains 251,872 events over 7,180 different drugs. The mean length of a sequence is 7.89$\pm$8.44 itemsets. Length variance is high due to the heterogenous nature of care pathways. Some of them represent complex therapies involving the consumption of many different drugs while others are simple case consisting of few deliveries of anti-epileptic drugs.

\espace

Let first compare results obtained by eNSP and \NegPSpan\ to illustrate the differences in the patterns sets extracted by each algorithm. To this end, we set up the algorithms with $\sigma=14.3\%$ ($1,200$ sequences), a maximum pattern length of $l=3$, $\tau=3$ for \NegPSpan\ and $\varsigma=.1\times\sigma$ the minimal support for positive partners for eNSP. eNSP extracts 1,120 patterns and \NegPSpan\ only 10 patterns (including positive and negative patterns). Due to a very low $\varsigma$ threshold, many positive patterns are extracted by eNSP leading to generate a lot of singleton negative patterns (\ie a pattern that hold a single negated item).

\begin{table}[tbh]
\caption{Patterns involving \textit{valproic acid} switches with their supports computed by eNSP and \NegPSpan.}
\small
\centering
\label{tab:valproicacid_comparison}
    \begin{tabular}{lcc}
        \hline
        pattern & \begin{tabular}[x]{@{}c@{}}support\\eNSP\end{tabular}  & \begin{tabular}[x]{@{}c@{}}support\\\NegPSpan\end{tabular}  \\ \hline
        $\seq{p}_1 = \langle 383\,\neg(86,383)\,383 \rangle$ & 1,579 & \\ 
        $\seq{p}_2 = \langle 383\,\neg 86\,383 \rangle$   & 1,251 & 1,243 \\ 
        $\seq{p}_3 = \langle 383\,\neg 112\,383 \rangle$  & 1,610 & \\ 
        $\seq{p}_4 = \langle 383\,\neg 114\,383 \rangle$  & 1,543 & 1,232\\
        $\seq{p}_5 = \langle 383\,\neg 115\,383 \rangle$  & 1,568 & 1,236\\
        $\seq{p}_6 = \langle 383\,\neg 151\,383 \rangle$  & 1,611 & \\
        $\seq{p}_7 = \langle 383\,\neg 158\,383 \rangle$  & 1,605 & \\
        $\seq{p}_8 = \langle 383\,\neg 7\,383 \rangle$    &      & 1,243\\
        \hline
    \end{tabular}
\end{table}

Precisely, we pay attention to the specific specialty of \textit{valproic acid} which exists in generic form (event $383$) or brand-named form (event $114$) by selecting patterns that start and finish with event $383$. The complete list of these patterns is given in Table \ref{tab:valproicacid_comparison}. 
Other events correspond to other anti-epileptic drugs 
 ($7$: \textit{levetiracetam}, $158$: \textit{phenobarbital})
or psycholeptic drugs 
 ($112$: \textit{zolpidem}, $115$: \textit{clobazam}, $151$: \textit{zopiclone})
except $86$ which is \textit{paracetamol}.

First, it is interesting to note that with this setting, the two algorithms share only 3 patterns $\seq{p}_2$, $\seq{p}_4$ and $\seq{p}_5$, which have lower support with \NegPSpan\ because of the maxgap constraint. This constraint also explains that pattern  $\seq{p}_3$ and $\seq{p}_6 $ are not extracted by \NegPSpan.
These patterns illustrate that in some cases, the patterns extracted by eNSP may not be really interesting because they involve distant events in the sequence.
Pattern $\seq{p}_{1}$ is not extracted by \NegPSpan\ due to the strict-embedding pattern semantics.
With eNSP semantics, $\seq{p}_{1}$ means that there is no delivery of \textit{paracetamol} and \textit{valproic acid} at the same time. With \NegPSpan\ semantics, $\seq{p}_{1}$ means that there is no delivery of \textit{paracetamol} neither \textit{valproic acid} between two deliveries of \textit{valproic acid}. The latter is stronger and the pattern support is lower.
On the opposite, \NegPSpan\ can extract patterns that are missed by eNSP. For instance, pattern $\seq{p}_8$ is not extracted by eNSP because its positive partner, $\langle 383, 7,383\rangle$, is not frequent. In this case, it leads eNSP to miss a potentially interesting pattern involving two anti-epileptic drugs.

Now, we look at patterns involving a switch from generic form to brand-named form of \textit{valproic acid} with the following settings $\sigma=1.2\%$, $l=3$ and $\tau=5$. Mining only positive patterns extracts the frequent patterns $\langle 114,383,114 \rangle$ and $\langle 114,114 \rangle$. It is impossible to conclude about the possible impact of a switch from $114$ to $383$ as a possible event triggering an epileptic crisis. From negative patterns extracted by \NegPSpan, we can observe that the absence of switch $\langle 114\,\neg 383\,114\rangle$ is also frequent in this dataset. Contrary to eNSP semantics which does bring a new information (that can be deduced from frequent patterns), this pattern concerns embeddings corresponding to real interesting cases thanks to gap constraints.


\section{Conclusion and Perspectives}\label{sec:ccl}
In this article, we investigated negative sequential pattern mining (NSP). It highlights that state of the art algorithms do not extract the same patterns, not only depending on their syntax and algorithms specificities, but also depending on the semantical choices. In this article, we have proposed definitions that clarify the negation semantics encountered in the literature. 
We have showed that NSP support depends on the semantics of itemset non-inclusion, two possible alternatives for considering negation of itemsets and two manners for considering multiple embeddings in a sequence.
This let us point out the limits of the state of the art algorithm eNSP that imposes a minimum support for positive partner and that is not able to deal with embedding constraints, and more especially maxgap constraints.

We have proposed \NegPSpan\ a new algorithm for mining negative sequential patterns that overcomes these limitations.
Our experiments show that \NegPSpan\ is more efficient than eNSP on datasets with medium long sequences (more than 20 itemsets) even when weak maxgap constraints are applied and that it prevents from missing possibly interesting patterns.

In addition, \NegPSpan\ is based on theoretical foundations that enable to extend it to the extraction of closed or maximal patterns to reduce the number of extracted patterns even more.

\subsection*{Acknowledgments}
  The authors would like to thank REPERES Team from Rennes University Hospital for spending time to discuss our case study results. 


\begin{thebibliography}{10}

\bibitem{cao2016nsp}
Longbing Cao, Xiangjun Dong, and Zhigang Zheng.
\newblock {e-NSP: Efficient negative sequential pattern mining}.
\newblock {\em Artificial Intelligence}, 235:156--182, 2016.

\bibitem{Cao2015}
Longbing Cao, Philip~S. Yu, and Vipin Kumar.
\newblock Nonoccurring behavior analytics: A new area.
\newblock {\em Intelligent Systems}, 30(6):4--11, 2015.

\bibitem{DauxaisAIME17}
Yann Dauxais, Thomas Guyet, David Gross{-}Amblard, and Andr{\'{e}} Happe.
\newblock Discriminant chronicles mining - application to care pathways
  analytics.
\newblock In {\em Proceedings of 16th Conference on Artificial Intelligence in
  Medicine}, volume 10259 of {\em Lecture Notes in Computer Science}, pages
  234--244. Springer, 2017.

\bibitem{dong2018f}
Xiangjun Dong, Yongshun Gong, and Longbing Cao.
\newblock {F-NSP+}: A fast negative sequential patterns mining method with
  self-adaptive data storage.
\newblock {\em Pattern Recognition}, 2018.

\bibitem{Gong2017eNSPFI}
Yongshun Gong, Tiantian Xu, Xiangjun Dong, and Guohua Lv.
\newblock {e-NSPFI}: Efficient mining negative sequential pattern from both
  frequent and infrequent positive sequential patterns.
\newblock {\em International Journal of Pattern Recognition and Artificial
  Intelligence}, 31(02):1750002, 2017.

\bibitem{hsueh:2008:PNSP}
Sue-Chen Hsueh, Ming-Yen Lin, and Chien-Liang Chen.
\newblock Mining negative sequential patterns for e-commerce recommendations.
\newblock In {\em Proceedings of Asia-Pacific Services Computing Conference},
  pages 1213--1218. IEEE, 2008.

\bibitem{Kamepalli:2014:FrequentNS}
Sujatha Kamepalli, Raja Sekhara, and Rao Kurra.
\newblock {Frequent Negative Sequential Patterns -- a Survey}.
\newblock {\em International Journal of Computer Engineering and Technology},
  5, 3:115--121, 2014.

\bibitem{lin2016fhn}
Jerry Chun-Wei Lin, Philippe Fournier-Viger, and Wensheng Gan.
\newblock {FHN}: An efficient algorithm for mining high-utility itemsets with
  negative unit profits.
\newblock {\em Knowledge-Based Systems}, 111:283--298, 2016.

\bibitem{liu2015sapnsp}
Chuanlu Liu, Xiangjun Dong, Caoyuan Li, and Yan Li.
\newblock Sapnsp: Select actionable positive and negative sequential patterns
  based on a contribution metric.
\newblock In {\em 12th International Conference on Fuzzy Systems and Knowledge
  Discovery}, pages 811--815. IEEE, 2015.

\bibitem{Mooney:2013}
Carl~H. Mooney and John~F. Roddick.
\newblock Sequential pattern mining -- approaches and algorithms.
\newblock {\em ACM Computing Survey}, 45(2):1--39, 2013.

\bibitem{Moulis2015411}
G.~Moulis, M.~Lapeyre-Mestre, A.~Palmaro, G.~Pugnet, J.-L. Montastruc, and
  L.~Sailler.
\newblock French health insurance databases: What interest for medical
  research?
\newblock {\em La Revue de M{\'e}decine Interne}, 36:411--417, 2015.

\bibitem{pei2004mining:prefixspan}
Jian Pei, Jiawei Han, Behzad Mortazavi-Asl, Jianyong Wang, Helen Pinto, Qiming
  Chen, Umeshwar Dayal, and Mei-Chun Hsu.
\newblock {Mining Sequential Patterns by Pattern-Growth: The PrefixSpan
  Approach}.
\newblock {\em IEEE Transactions on knowledge and data engineering},
  16(11):1424--1440, 2004.

\bibitem{polard2015brand}
Elisabeth Polard, Emmanuel Nowak, Andr{\'e} Happe, Arnaud Biraben, and Emmanuel
  Oger.
\newblock Brand name to generic substitution of antiepileptic drugs does not
  lead to seizure-related hospitalization: a population-based case-crossover
  study.
\newblock {\em Pharmacoepidemiology and drug safety}, 24:1161--1169, 2015.

\bibitem{qiu2017negi}
Ping Qiu, Long Zhao, and Xiangjun Dong.
\newblock {NegI-NSP}: Negative sequential pattern mining based on loose
  constraints.
\newblock In {\em 43rd Annual Conference of the IECON}, pages 3419--3425. IEEE,
  2017.

\bibitem{srikant1996:GSP}
Ramakrishnan Srikant and Rakesh Agrawal.
\newblock {Mining sequential patterns: Generalizations and performance
  improvements}.
\newblock In {\em International Conference on Extending Database Technology},
  pages 1--17. Springer, 1996.

\bibitem{xu2017HighUtilityNSP}
Tiantian Xu, Xiangjun Dong, Jianliang Xu, and Xue Dong.
\newblock Mining high utility sequential patterns with negative item values.
\newblock {\em International Journal of Pattern Recognition and Artificial
  Intelligence}, 31(10):1750035, 2017.

\bibitem{xu2017msnsp}
Tiantian Xu, Xiangjun Dong, Jianliang Xu, and Yongshun Gong.
\newblock {E-msNSP}: Efficient negative sequential patterns mining based on
  multiple minimum supports.
\newblock {\em International Journal of Pattern Recognition and Artificial
  Intelligence}, 31(02):1750003, 2017.

\bibitem{xu2018efficient}
Tiantian Xu, Tongxuan Li, and Xiangjun Dong.
\newblock Efficient high utility negative sequential patterns mining in smart
  campus.
\newblock {\em IEEE Access}, 6:23839--23847, 2018.

\bibitem{zheng:2009:negative}
Zhigang Zheng, Yanchang Zhao, Ziye Zuo, and Longbing Cao.
\newblock {Negative-GSP: An efficient method for mining negative sequential
  patterns}.
\newblock In {\em Proceedings of the Australasian Data Mining Conference},
  pages 63--67, 2009.

\bibitem{zheng2010efficient}
Zhigang Zheng, Yanchang Zhao, Ziye Zuo, and Longbing Cao.
\newblock An efficient {GA}-based algorithm for mining negative sequential
  patterns.
\newblock In {\em Pacific-Asia Conference on Knowledge Discovery and Data
  Mining}, pages 262--273. Springer, 2010.

\end{thebibliography}

\newpage
\appendix
\section{Proofs}\label{sec:proofs}

\begin{proof}[Proof of Proposition \ref{prop:sqsubset_eqembeddings}]
Let $\seq{s}=\langle s_1,\dots, s_n\rangle$ be a sequence and $\seq{p}=\langle p_1,\dots, p_m\rangle$ be a negative sequential pattern.
Let $\seq{e}=(e_i)_{i\in[m]}\in [n]^m$ be a soft-embedding of pattern $\seq{p}$ in sequence $\seq{s}$. Then, the definition matches the one for strict-embedding if $p_i$ is positive. If $p_i$ is negative then $\forall j\in [e_{i-1}+1,e_{i+1}-1],\; p_i \not\sqsubseteq s_j$, \ie $\forall j\in [e_{i-1}+1,e_{i+1}-1],\; \forall \alpha\in p_i,\; \alpha \notin s_j$ and then $\forall \alpha\in p_i,\; \forall j\in [e_{i-1}+1,e_{i+1}-1],\; \alpha \notin s_j$.
It thus implies that $\forall \alpha\in p_i,\; \alpha \notin \bigcup_{j\in [e_{i-1}+1,e_{i+1}-1]} s_j$, \ie by definition, $p_i \not\sqsubseteq \bigcup_{j\in [e_{i-1}+1,e_{i+1}-1]} s_j$.

The exact same reasoning is done in reverse way to prove the equivalence.
\end{proof}


\begin{proof}[Proof of Proposition \ref{prop:antimonotonic} (Anti-monotonicity of NSP)]
Let $\seq{p} = \langle p_1\ \neg q_1\ p_2\ \neg q_2\ \dots p_{k-1}\ \neg q_{k-1}\ p_{k}\rangle$ and $\seq{p}' =\langle p'_1\ \neg q'_1\ p'_2\ \neg q'_2\ \dots p'_{k'-1}\ \neg q'_{k'-1}\ p'_{k'}\rangle$ be two NSP s.t. $\seq{p}\lhd \seq{p}'$.
And let $\seq{s}=\langle s_1,\dots,s_n\rangle$ be a sequence s.t. $ \seq{p}'\preceq \seq{s}$, \ie it exists an embedding $(e_i)_{i\in[k']}$:
\begin{itemize}
\item $\forall i,\; e_{i+1}>e_i$ (embedding), $e_{i+1}-e_i\leq\theta$ (\textit{maxgap}) and $e_{k'}-e_1\leq\tau$ (\textit{maxspan}),
\item $\forall i,\; p'_i\subseteq s_{e_i}$,
\item $\forall j\in [e_i+1,e_{i+1}-1],\; q'_i\not\sqsubseteq s_{e_j}$
\end{itemize}

To prove that $\seq{p}\preceq \seq{s}$, we prove that $(e_i)_{i\in[k]}$ is an embedding of $ \seq{p}$ in $\seq{s}$.

Let us first consider that $k=k'$, then by definitions of $\lhd$ and the embedding,
\begin{enumerate}[(i)]
\item $\forall i\in[k],\; p_i\subseteq p'_i\subseteq s_{e_i}$,
\item $\forall i\in[k-1],\;\forall j\in [e_i+1,e_{e+1}-1],\; q'_j\not\sqsubseteq s_{e_i}$, and thus $q_j\not\sqsubseteq s_{e_i}$ (because of anti-monotonicity of $\not\sqsubseteq$ and $q_i\subseteq q'i$)
\end{enumerate}
In addition, \textit{maxgap} and \textit{maxspan} constraints are satisfied by the embedding, \ie
\begin{enumerate}[(i)]\setcounter{enumi}{3}
\item $\forall i\in[k],\;e_{i+1}-e_i\leq\theta$
\item $e_{k}-e_1=e_{k'}-e_1\leq\tau$
\end{enumerate}
This means that $(e_i)_{i\in[k]}$ is an embedding of $ \seq{p}$ in $\seq{s}$.

\medskip

Let us now consider that $k'>k$, $(i)$, $(ii)$ and $(iii)$ still holds, and we have in addition that $e_{k}<e_{k'}$ (embedding property), then $e_{k}-e_i<\theta$.
This means that $(e_i)_{i\in[k]}$ is an embedding of $\seq{p}$ in $\seq{s}$.
\end{proof}

\begin{proof}[Proof of proposition \ref{prop:complet_correct} (Complete and correct algorithm)]
The correction of the algorithm is given by lines 2-3 of Algorithm \ref{algo:CeNSP-Rec}. A pattern is outputted only if it is frequent (line 2).

We now prove the completeness of the algorithm. 
First of all, we have to prove that any pattern can be reached using a path of elementary transformations ($\leadsto \in \{\leadsto_n, \leadsto_s, \leadsto_c\}$). 
Let $\seq{p}'=\langle p'_1 \dots p'_{m}\rangle$ be a pattern with a total amount of $n$ items, $n>0$, then it is possible to define $\seq{p}$ such that  $\seq{p} \leadsto \seq{p}'$ where $\leadsto \in \{\leadsto_n, \leadsto_s, \leadsto_c\}$, and $\seq{p}$ will have exactly $n-1$ items:
\begin{itemize}
\item if the last itemset of $\seq{p}'$ is such that $|p'_{m}|>1$ we define $\seq{p}=\langle p'_1 \dots p'_{m-1}\ p_{m}\rangle$ as the pattern with the same prefix as $\seq{p}'$ and an additional itemset, $p_{m}$ such that $|p_{m}|=|p_{m}'|-1$ and $p_{m} \subset p_{m}'$: then $\seq{p} \leadsto_c \seq{p}'$
\item if the last itemset of $\seq{p}'$ is such that $|p'_{m}|==1$ and $p'_{m-1}$ is positive then we define $\seq{p}=\langle p'_1 \dots p'_{m-2}\ p'{m}\rangle$: then $\seq{p} \leadsto_s \seq{p}'$
\item if the last itemset of $\seq{p}'$ is such that $|p'_{m}|==1$ and $p'_{m-1}$ is negative (non-empty) then we define $\seq{p}=\langle p'_1 \dots p_{m-1}\ p'{m}\rangle$ where $p_{m-1}$ is such that $|p_{m-1}|=|p_{m-1}'|-1$ and $p_{m-1} \subset p_{m-1}'$: then $\seq{p} \leadsto_n \seq{p}'$
\end{itemize}
Applying recursively this rules we have that for any pattern $\seq{p}$ there is a path from the empty sequence to it: $\emptyset \leadsto^* \seq{p}$. 
We can also notice that there is only one possibility between the three extensions, meaning that these path is unique. This prove that our algorithm is not redundant.

\medskip

Second, the pruning strategy is correct such that any frequent pattern will be missed. It is given by the anti-monotonicity property. 

Let $\seq{p}$ and $\seq{p}'$ be two patterns such that $\seq{p} \leadsto \seq{p}'$ where $\leadsto \in \{\leadsto_n, \leadsto_s, \leadsto_c\}$, then is is quite obvious that $\seq{p} \lhd \seq{p}'$.
Let's now consider that $\seq{p} \leadsto^* \seq{p}'$ from $\seq{p}$ to $\seq{p}'$ then, by transitivity of $\lhd$, we also have that $\seq{p} \lhd \seq{p}'$. And then by anti-monotonicity of the support, we have that $supp(\seq{p}) \geq supp(\seq{p}') $.

Let us now proceed by contradiction and consider that $\seq{p}'$ is a pattern such that $supp(\seq{p}')\geq \sigma$ but that the algorithm didn't find out. This means that for all paths\footnote{Note that we proved that this path is actually unique.} $\emptyset \leadsto^* \seq{p}'$ there exist $\seq{p}$ such that $\emptyset \leadsto^* \seq{p} \leadsto^* \seq{p}'$ with $supp(\seq{p})<\sigma$. $\seq{p}$ the pattern that has been used to prune the search exploration of this path to $\seq{p}'$.
This is not possible considering that $\seq{p} \leadsto^* \seq{p}'$ and thus that $supp(\seq{p}) \geq supp(\seq{p}') \geq \sigma$.
\end{proof}

\section{NegPSpan extracts a superset of eNSP}\label{sec:proofs_superset}

\begin{proposition}[]\label{prop:soft_implies_strict}
Soft-embedding $\implies$ strict-embedding for patterns consisting of items.
\end{proposition}

\begin{proof}
Let $\seq{s}=\langle s_1,\dots,s_n\rangle$ be a sequence and $\seq{p}=\langle p_1,\dots,p_m\rangle$ be a NSP s.t. each $\forall i,\; |p_i|=1$ and $\seq{p}$ occurs in $\seq{s}$ according to the soft-embedding semantic.

There exists $\epsilon =(e_i)_{i\in[m]}\in[n]^m$ s.t. for all $i\in[n]$,  $p_i$ is positive implies $p_i\in s_{e_i}$ and $p_i$ is negative implies that for all $j\in [e_{i-1}+1,e_{e+1}-1],\; p_i\notin s_j$ (items only) then $p_i\notin \bigcup_{j\in [e_{i-1}+1,e_{e+1}-1]}{s_j}$ \ie $p_i \not\subseteq \bigcup_{j\in [e_{i-1}+1,e_{e+1}-1]}{s_j}$ (no matter $\not\preceq$ or $\not\sqsubset$). As a consequence $\epsilon$ is a strict-embedding of $p$.
\end{proof}

\begin{proposition}\label{prop:eNSP_implies_CeNSP}
Let $\mathcal{D}$ be a dataset of sequences of items and $\seq{p}=\langle p_1,\dots,p_m\rangle$ be a sequential pattern extracted by eNSP, then without embedding constraints $\seq{p}$ is extracted by \NegPSpan\ with the same minimum support.
\end{proposition}

\begin{proof} 
If $\seq{p}$ is extracted by eNSP, it implies that its positive partner is frequent in the dataset $\mathcal{D}$. As a consequence, each $p_i$, $i\in[m]$ is a singleton itemset.

According to the search space of \NegPSpan\ defined by $\lhd$ if $\seq{p}$ is frequent then it will be reached by the depth-first search.
Then it is sufficient to prove that for any sequence $\seq{s}=\langle s_1,\dots,s_n\rangle \in \mathcal{D}$ such that $\seq{p}$ occurs in $\seq{s}$ according to eNSP semantic (strict-embedding, strong absence), then $\seq{p}$ also occurs in $\seq{s}$ according to the \NegPSpan\ semantics (soft-embedding, weak absence). With that and considering the same minimum support threshold, $\seq{p}$ is frequent according to \NegPSpan.
Proposition \ref{prop:soft_implies_strict} gives this result.
\end{proof}

Then we conclude that \NegPSpan\ extracts more patterns than eNSP on sequences of items. In fact, \NegPSpan\ can extract patterns with negative itemsets larger than 2.

eNSP extract patterns that are not extracted by \NegPSpan\ on sequences of itemsets.
Practically, \NegPSpan\ uses a size limit for negative itemsets $\nu\geq 1$.
eNSP extracts patterns whose positive partners are frequent. The positive partner, extracted by PrefixSpan may hold itemsets  larger than $\nu$, and if the pattern with negated itemset is also frequent, then this pattern will be extract by eNSP, but not by \NegPSpan.

\end{document}